\documentclass[acmsmall]{acmart}
\AtBeginDocument{%
  }

\setcopyright{acmlicensed}
\copyrightyear{2018}
\acmYear{2018}
\acmDOI{XXXXXXX.XXXXXXX}

\acmJournal{JACM}
\acmVolume{37}
\acmNumber{4}
\acmArticle{111}
\acmMonth{8}

\usepackage{amsthm}
\usepackage{algorithmic}
\usepackage[linesnumbered, vlined, ruled]{algorithm2e}
\usepackage{graphicx}
\graphicspath{{fig}{./fig/}}
\usepackage{subcaption}
\usepackage{soul} 
\usepackage[dvipsnames]{xcolor}

\newcommand{\annotate}[1]{#1}

\setstcolor{red}

\SetCommentSty{mycommfont}

\begin{document}

\title{Bounded Priority-Aware Locking for Real-Time Kernels}

\author{Shriram Raja}
\email{shriramr@bu.edu}
\orcid{0000-0002-3123-3486}
\author{Richard West}
\email{richwest@bu.edu}
\orcid{0000-0001-5100-0666}
\affiliation{%
  \institution{Department of Computer Science, Boston University}
  \city{Boston}
  \state{Massachusetts}
  \country{USA}
}


\begin{abstract}
 A real-time multicore system requires delay bounds on access to
 shared resources. These resources include the kernel, which has
 potentially many non-preemptible critical sections guarded by one or more
 different synchronization primitives. While primitives such as FIFO locks bound
 the waiting time to enter a critical section, they do not distinguish
 the importance of individual tasks competing for shared resource
 access. To address this, we consider a priority-aware spinlock, which
 reduces the average delay of more important tasks while maintaining a
 worst-case bound on lock waiting time.

 We propose a Batched Priority Lock (BPL) that first groups waiting
 tasks based on the order of their lock requests, and then determines the
 next lock holder according to priority within the waiting group.  We
 compare BPL to alternative lock approaches, showing that the average waiting
 time is reduced for higher priority tasks, in simulations up to 64 cores, and
 for a working implementation on an 8-core machine with a real RTOS. BPL is a
 compromise between strict priority and FIFO ordering. While strict
 priorities may lead to starvation and, hence, unbounded lock
 acquisition delays, BPL has the same waiting bound as FIFO, but with
 benefits to higher priority tasks. Although its complexity is greater
 than that of a simple spinlock, its common case execution overhead is
 shown to be inexpensive in a working system. We believe this is 
 an acceptable cost in systems that require predictability.
\end{abstract}




\maketitle

\section{Introduction}\label{sect:introduction}
Synchronized access to shared resources is a fundamental problem faced
by operating systems. It is especially challenging in real-time
systems, which must meet strict timing requirements. In general, the
worst-case waiting time to exclusively access a shared resource
increases with the number of concurrent tasks competing for that
resource, which leads to potential deadline misses. Moreover,
important tasks might be delayed by those having less importance as a
result of priority inversion. A poorly designed locking mechanism only
exacerbates these problems, motivating us to consider practical
solutions that build upon the state-of-the-art. This paper, therefore,
focuses on the design of spinlocks that ensure bounded delay and
minimize priority inversion. Although spinlocks are generally
applicable to multicore systems, we primarily consider access to
non-preemptible shared kernel control paths and data structures, by
entry points such as system calls and interrupts in a real-time
system.

An {\em unordered} spinlock is the simplest mechanism to serialize accesses
to a shared resource. Only one task is allowed to change the lock
status and access the resource, while all other tasks busy wait until
the lock is free. Highly contended resources increase the likelihood
of multiple waiters. In such cases, it is not clear which waiter
acquires the spinlock on its release, leading to variable and
potentially high wait times. 

In contrast, FIFO-ordered locks have been used to bound the waiting time
of lock contenders~\cite{mcs_lock, clh_c,clh_lh}. When a resource is shared by
all cores in an $m$-core system, FIFO locks with non-preemptible critical
sections provide a worst-case waiting time of $m-1$ times the worst-case
duration of the critical section. Since FIFO locks only consider order of
arrival, the average waiting time of tasks of all priorities is the same. 

Typically, worst-case delay is the only factor considered when
determining if all tasks meet their deadlines. However, there is value
in improving the average-case performance of higher priority tasks, to
allow them to make more progress or finish sooner. Allowing more
important tasks to complete earlier, or perform additional work beyond
their minimal execution requirements, may lead to improved system-wide
quality-of-service. For example, a task that performs a numerical
integration might operate at a mandatory worst-case sampling rate to
ensure minimal accuracy, with any additional cycles used for
over-sampling, to optionally reduce integration error. This is similar
to the work on imprecise computations~\cite{Liu:91,Liu:94}. To that end, we
propose the use of a priority-aware lock that provides the same
worst-case bound as FIFO but prioritizes among tasks to improve the
average case performance of more important tasks.

Several priority queue-based locks have previously been
proposed~\cite{bburg_lockreview}, which support preemption of waiters and
timeouts. However, priority queue-based locks either: (1) require the releasing
task to traverse a queue to determine the next lock holder, extending the
critical section, or (2) incur queue insertion delays during lock acquisition,
which cause unbounded waiting in pathological
cases~\cite{priority_mutual_exclusion}.
To avoid queue insertions of waiting tasks at acquisition time, while
still ensuring a constant release time, we propose a
two-stage \emph{Batched Priority Lock} (BPL): first, waiting tasks are
batched in order of arrival, and then contend to resolve the waiter
with the highest priority among those in the same batch. A strictly
priority-ordered lock may result in starvation of low priority
tasks. However, by first considering the order of lock request, and
then the priority, BPL ensures progress of all waiters. Any tasks that
make a request for a lock held by the same holder are batched
together, allowing precedence to be given to the highest priority
waiter in that group when the lock is released.

{\bf Contributions:} We describe the implementation of BPL, and
compare against FIFO, strictly priority-ordered, and simple
spinlocks. Using simulations for systems up to 64 cores, and
implementations of the different locks in our in-house Real-Time Operating System Quest,
we determine: (1)
the overheads of each lock, (2) their effectiveness at eliminating
priority inversions, and (3) the conditions under which BPL
proves more advantageous than other locks.

The rest of this paper is organized as follows:
Section~\ref{sect:resource_share} explains the issues with using
priority-unaware spinlocks in a multicore real-time
system. Section~\ref{sect:kernel_latency} gives an overview of our
real-time kernel and lists the desired properties of a kernel
lock. The design of the Batched Priority Lock is described in
Section~\ref{sect:lock} and analyzed in Section~\ref{sect:analysis}.
Section~\ref{sect:exp_eval} evaluates different lock
approaches. Related work is discussed in Section~\ref{sect:rel_work},
followed by conclusions in Section~\ref{sect:conclusion}.

\section{Resource Sharing in a Multicore RTOS}\label{sect:resource_share}
We consider a multi-core real-time system with $m$ cores where tasks are
scheduled using globally consistent, statically assigned priorities.
One
instance of priority inversion is said to occur in a real-time system when a
task is unable to make progress due to the execution of one lower priority task.
Multicore real-time systems with globally consistent priorities suffer from: (1)
\emph{local} priority inversion between tasks on the same core, and (2)
\emph{global} priority inversions between tasks on different cores. Unbounded
local priority inversion is avoided by ensuring that critical sections are
non-preemptible, either by boosting the priority of the lock holder to the
highest priority in the system for the duration of the critical section, or by
disabling interrupts. However, even non-preemptible critical sections guarded by
a simple spinlock or FIFO-ordered lock may suffer global priority inversion as
these locks do not give precedence to higher priority tasks as illustrated by
the following example.

\begin{figure}[!htb]
  \centering
  
  \includegraphics[width=0.5\linewidth]{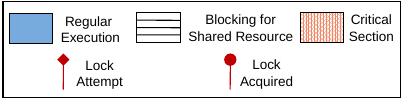}
  
  \vspace{0.3cm}
  
  \begin{subfigure}[b]{0.48\linewidth}
    \centering
    \includegraphics[width=\linewidth]{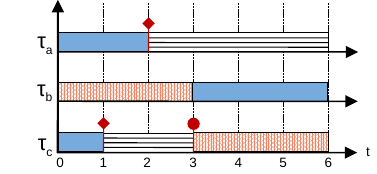}
    \caption{Priority Inversion due to non-priority-ordered locks: $\tau_a$
    experiences one instance of priority inversion from $t = 2$ till $\tau_c$
    completes its critical section}
    \label{fig:p_unaware_lock}
  \end{subfigure}
  \hfill
  \begin{subfigure}[b]{0.48\linewidth}
    \centering
    \includegraphics[width=\linewidth]{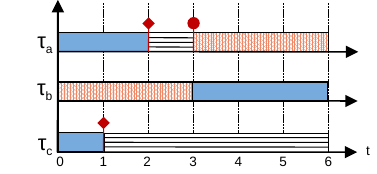}
    \caption{Priority Inversion avoided with priority-aware lock: $\tau_a$
    acquires the lock at $t = 3$ \newline }
    \label{fig:p_aware_lock}
  \end{subfigure}
  
  \caption{Example of Global Priority Inversion in Multicore Real-Time Systems}
  \label{fig:p_inv}
  \Description{Example of Global Priority Inversion in Multicore Real-Time Systems}
\end{figure}

{\bf Example:} Consider a 3-core system shown in
Figure~\ref{fig:p_unaware_lock}. 
Suppose task $\tau_a$ has the
highest priority, $\tau_c$ the lowest priority, and $\tau_b$ the
middle priority, with all tasks assigned to different cores. At time
$t = 1$, $\tau_b$ is in its critical section. $\tau_c$ and $\tau_a$
attempt to acquire a shared FIFO lock at times $t = 1$ and $t = 2$,
respectively. $\tau_b$ exits its critical section at time $t = 3$, but
since the FIFO lock does not enforce
priority ordering, $\tau_c$ acquires the lock, resulting in the
highest priority task, $\tau_a$, waiting at least the length of
another critical section. If a priority-aware lock were used, as shown
in Figure~\ref{fig:p_aware_lock}, the highest priority task $\tau_a$ will
acquire the lock at $t = 2$. Extending this situation to
a lock contended by $m$ tasks running on separate cores of an $m$-core system, a
FIFO lock will cause the last arriving waiter to be delayed by $m-1$ other
tasks, irrespective of their priority. The challenge, then, is to
implement a priority-based lock that reduces the waiting time of the most
important tasks, but also ensures a bounded delay for all waiters. 

There exist several priority-aware locks that use queue data structures to
ensure that only the highest priority contender is allowed to acquire the lock
when it is released~\cite{clh_c,clh_lh,markatos_plock,o1_spinlock}. These locks
are broadly classified into two categories:

\begin{itemize}
\item {\em Release-prioritized locks}: this category of locks uses a simple FIFO
queue for waiters that attempt to acquire a lock. However, when the lock holder
completes its critical section, it traverses the entire queue to identify the
highest priority waiter, and passes ownership of the lock to that task. While
this approach is simpler to implement, it effectively extends each critical
section by the time taken to determine the next highest priority waiter and
atomically remove it from the FIFO queue. This extra cost is part of the lock
release operation by the current holder.

\item {\em Acquire-prioritized locks}: these use a priority queue to order
waiters as they attempt to acquire a lock. The releasing task is then able to
quickly identify the next task to take lock ownership, avoiding costly delays
being added to its critical section. In effect, all waiters cooperate to
identify the highest priority task while the lock is in use. However,
complications arise in the organization of the priority queue, which is
performed by a collection of concurrent waiters. Huang and
Jayanti~\cite{priority_mutual_exclusion} identify situations where a task fails
to enqueue itself into the priority queue due to repeated failed transactions.
\end{itemize}

David et al.~\cite{everything_synch} conduct a thorough study on
synchronization, and show that simple locks are often preferred even in cases of
high contention. In contrast, complex locks often incur costs without providing
net benefits. We postulate that a lock designed for a real-time operating system
must:
\begin{itemize}
  \item ensure mutual exclusion with minimal overhead, at least in the common
  case, using widely available hardware instructions,
  \item minimize the waiting time of high priority tasks by bounding priority
  inversion, and
  \item provide progress guarantees to avoid starvation of low priority tasks.
\end{itemize}

\section{Kernel Latency in a Multicore RTOS}\label{sect:kernel_latency}
Symmetric multiprocessing (SMP) features two or more processors, or
cores, which are connected to a shared main memory managed by a single
operating system. The trusted {\em kernel} component of such an
operating system is considered a shared resource that is accessed by
tasks on different cores using system calls. These system calls
require access to synchronization primitives, to exclusively update
shared kernel state or avoid contention on access to operating system
resources that must be serialized. Counting semaphores, mutexes,
blocking and spinning locks are commonly used synchronization
primitives. Such primitives may be used to guard access to the entire
kernel (e.g., a big kernel lock) or may be used to implement mutually
exclusive access to small regions of kernel code (critical sections)
guarded by separate fine-grained locks. Regardless of coarse or
fine-grained locking, it remains a challenge to avoid unbounded wait
time when acquiring access to shared resources, while also giving
precedence to higher priority tasks.

In this work, we consider Quest RTOS as a case study. Scheduling in Quest is
done using the concept of bandwidth preserving sporadic
servers~\cite{sporadic_server} implemented using the Virtual CPU (VCPU)
abstraction proposed by Danish et al~\cite{vcpu}. One or more tasks are
assigned to a given VCPU, having an execution budget $C$ every period $T$. Each
per-core local scheduler selects the highest priority runnable VCPU, having a
non-zero budget at the current time. The next task to run using this VCPU is
dispatched on the local processor core. The VCPUs are partitioned across cores,
but are allowed to migrate for load balancing purposes. On each core, the
Rate-Monotonic algorithm~\cite{rms} is implemented to schedule the VCPUs.
Either FIFO or some form of priority-based scheduling is used to select the
next task to run on a VCPU shared by more than one task. The Liu-Layland
test~\cite{rms} is used to verify the schedulability of the VCPUs assigned to a
core. Hence, Quest can be said to use Semi-Partitioned Fixed-Priority
scheduling. Similar to UNIX-like systems, Quest uses segmentation to separate
the kernel from user-space.

\begin{figure*}[!ht]
    \centering

    \includegraphics[width=0.48\textwidth]{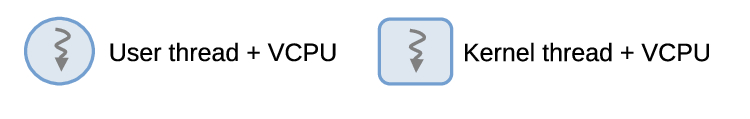}

    \begin{subfigure}{0.48\textwidth}
        \includegraphics[width=0.96\linewidth]{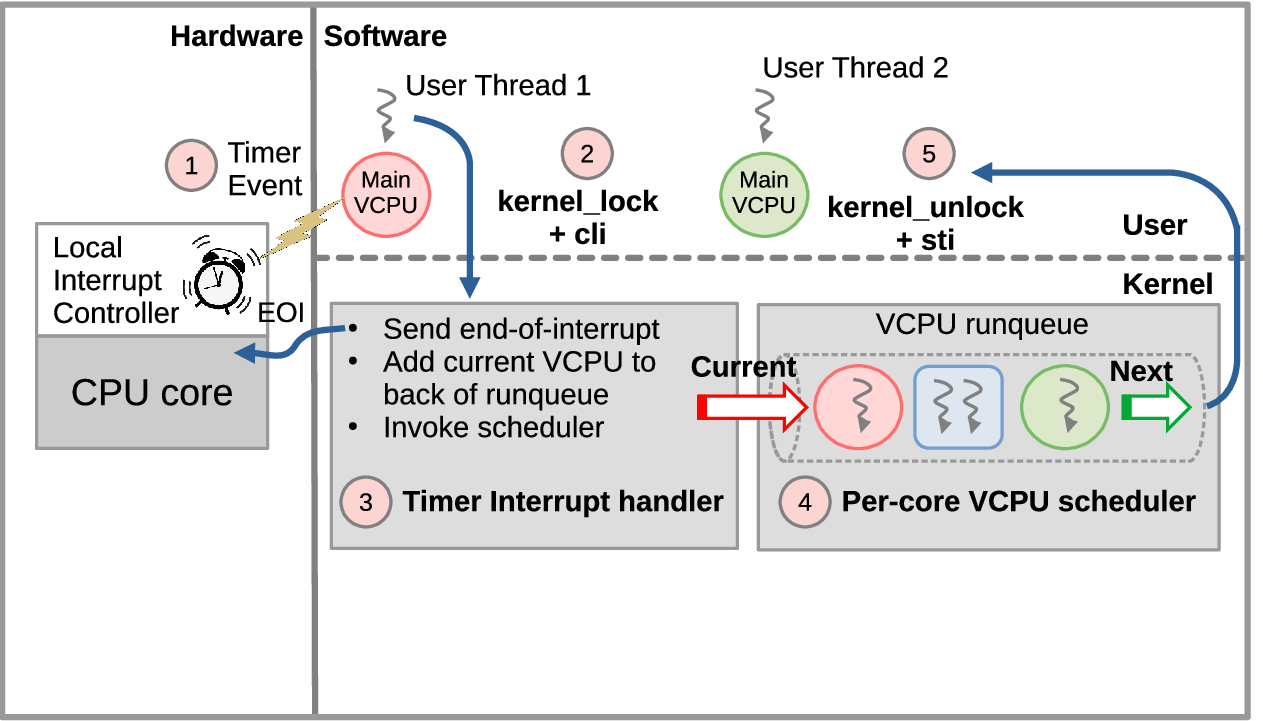}
        \caption{User to User Task Transition}
        \label{fig:u2u}
    \end{subfigure}
    \begin{subfigure}{0.48\textwidth}
        \includegraphics[width=0.96\linewidth]{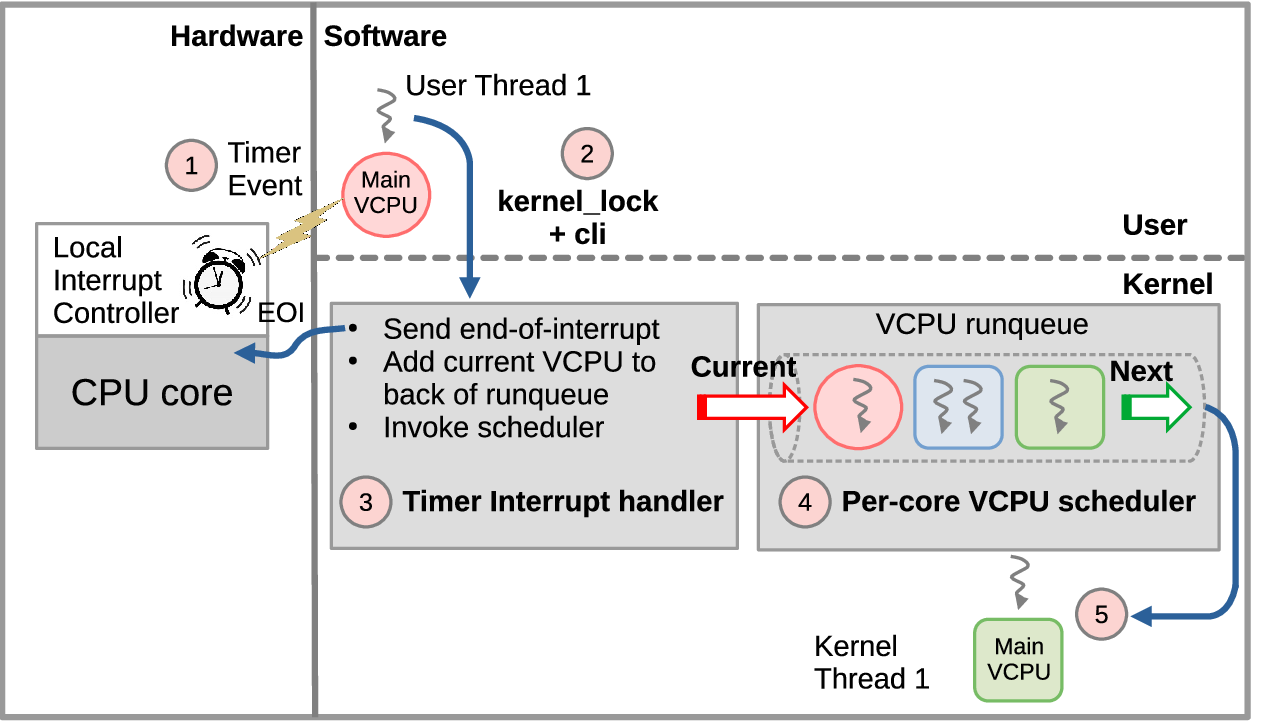}
        \caption{User to Kernel Task Transition}
        \label{fig:u2k}
    \end{subfigure}

    \begin{subfigure}{0.48\textwidth}
        \includegraphics[width=0.96\linewidth]{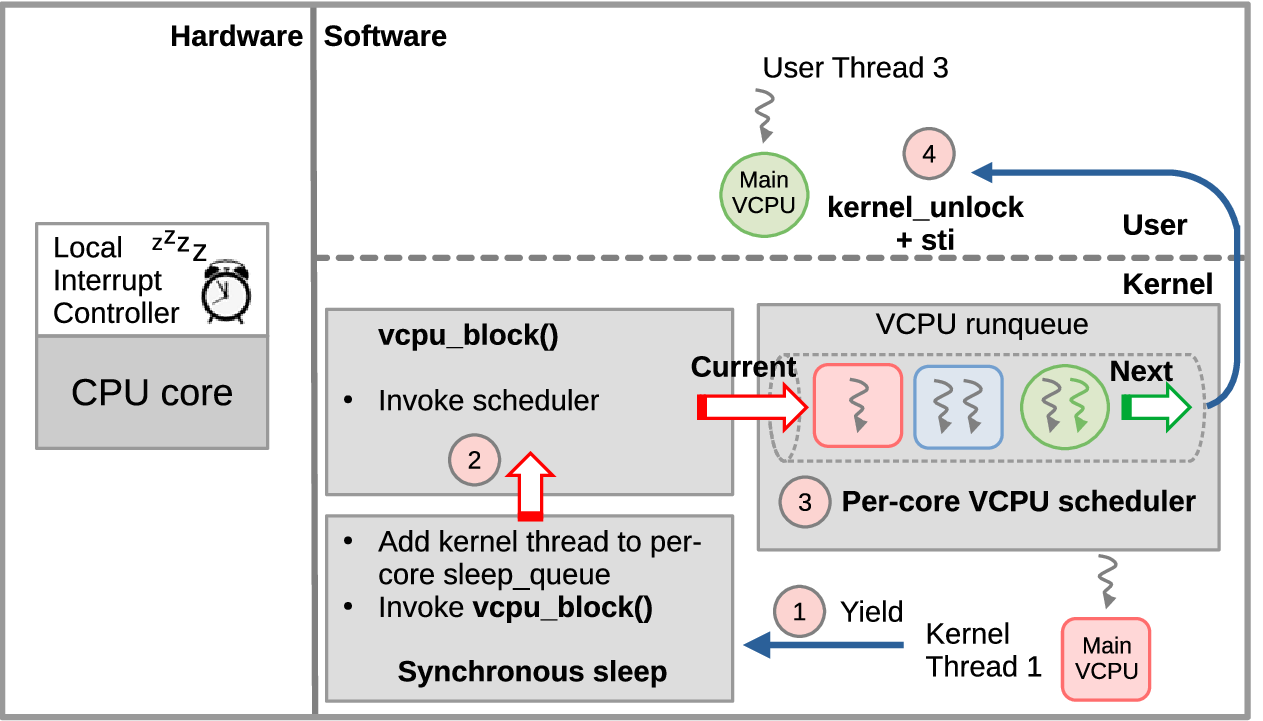}
        \caption{Kernel to User Task Transition}
        \label{fig:k2u}
    \end{subfigure}
    \begin{subfigure}{0.48\textwidth}
        \includegraphics[width=0.96\linewidth]{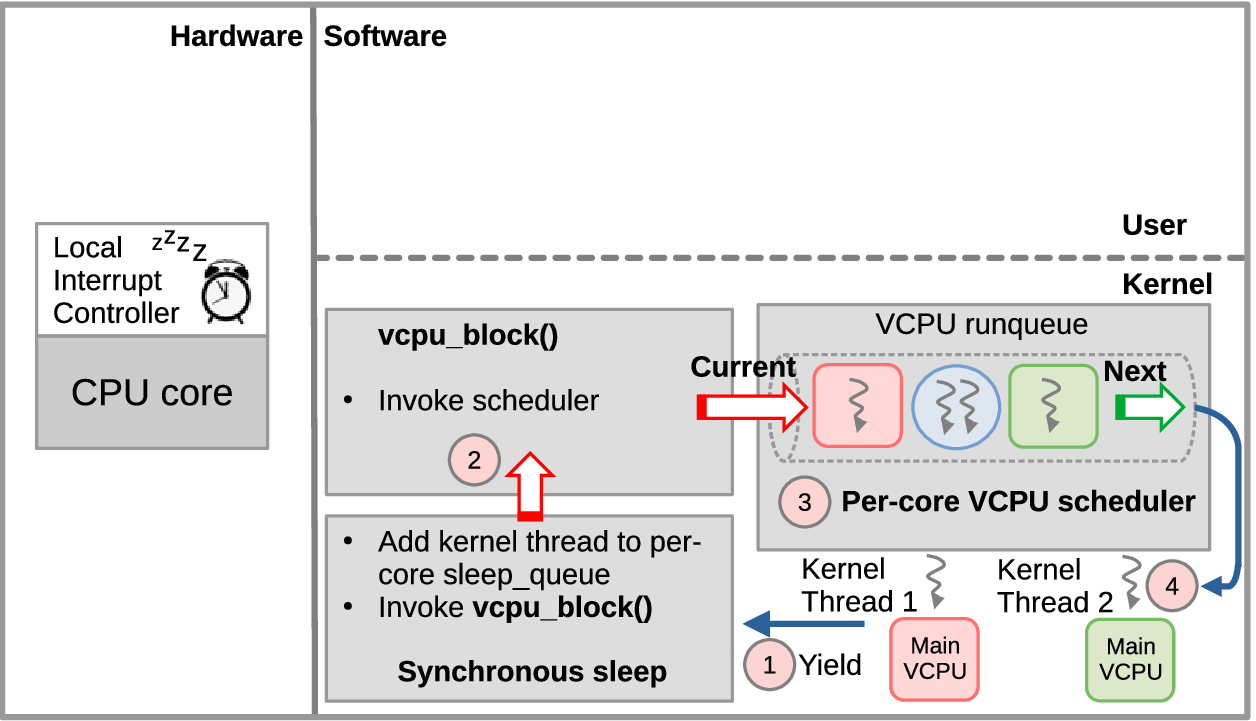}
        \caption{Kernel to Kernel Task Transition}
        \label{fig:k2k}
    \end{subfigure}

    \caption{Scheduler Paths in Quest RTOS}
    \label{fig:rtk_kernel_paths}
    \Description{Scheduler Paths in Quest RTOS}
\end{figure*}

\annotate{Figure~\ref{fig:rtk_kernel_paths} shows the different control
paths through the kernel:}

\annotate{\textbf{User Task to User Task} (Figure ~\ref{fig:u2u}): User-space tasks are
interrupted by the local APIC (LAPIC) timer, either due to preemption or the
expiration of the corresponding VCPU budget. Further interrupts are
disabled before control passes to the handler. The handler first sends an
end-of-interrupt signal to the timer and then acquires the kernel lock. After
processing the sleepqueue of the system, the current task is placed in the back
of the core-local (per-cpu) run queue and then the scheduler is called. The
scheduler chooses the highest priority VCPU with a non-zero budget, and
control is switched to the next task in the chosen VCPU.  Finally, the kernel
lock is released and interrupts are re-enabled as control passes to the
user-space.}

\annotate{\textbf{User Task to Kernel Task} (Figure ~\ref{fig:u2k}): Similar to the
previous case, the user task is interrupted, leading to control passing to the
LAPIC timer handler. The only difference is, in this case the VCPU chosen to run
next has a kernel task mapped to it. Hence, the kernel task immediately proceeds
while still holding the kernel lock and with interrupts disabled.}

\annotate{\textbf{Kernel Task to User Task} (Figure ~\ref{fig:k2u}): Since kernel tasks
run with interrupts disabled, control switches to another task only when a
kernel task yields to the scheduler using a {\tt vcpu\_block()} or a
{\tt sleep()} call. If a user-space task is mapped to the next runnable VCPU, the
kernel lock is released and interrupts are re-enabled.}

\annotate{\textbf{Kernel Task to Kernel Task} (Figure ~\ref{fig:k2k}): When switching from
one kernel task to another, the interrupts remain disabled and the kernel lock
is held throughout, as the flow of control passes from the first kernel task, to
the scheduler and finally to the next kernel task.}

\annotate{For cases where a user task issues a blocking system call, control will first pass to the kernel, acquiring the kernel lock and disabling interrupts. A subsequent switch to either a kernel or user task from this point is similar to the control flow originating from a kernel task. }

Quest uses a big kernel lock to enforce mutual exclusion in critical
kernel code. This simplifies system design and avoids potential
``hold and wait'' deadlocks caused by fine-grained locking of
different kernel control paths guarded by separate synchronization
constructs. While fine-grained locking may improve scalability, it
also introduces the following problems:\begin{itemize}
    \item \textbf{Formal Verification of Correctness.} In their survey on Microsoft
platforms that use concurrency~\cite{ms_conc_survey}, Godefroid and Nagappan
found that 66\% of the 684 respondents dealt with concurrency issues.
While this study does not focus on real-time kernels, use of
several fine-grained locks in the kernel also makes the paths through the kernel
more complex. This in turn makes formal verification of system correctness very
complicated. 
    \item \textbf{Worst-Case Timing Analysis.} When the entire kernel is
considered as a single critical section, the worst-case timing analysis of tasks
is simplified. In contrast, separate locks within the kernel leads to
potentially interleaved, or nested, critical sections, adding to the complexity
of worst-case timing analysis.
\end{itemize}

There is a trade-off to be made here between scalability and system
complexity.  Also, it must be noted that predictability is of higher
importance than scalability for a real-time system. In a system with
$m$ cores, the maximum number of non-preemptible threads that could
contend for the kernel lock is limited to $m$. Thus, when $m$ is
generally small, the level of contention is lower. Even if a real-time
system were to require many cores, it is possible to split the cores
into separate groups, e.g., by using multiple guests running on a
paritioning hypervisor~\cite{quest-v,bao}. Only one group of cores
assigned to a guest would then need to contend for the same lock. \annotate{This is
desirable in mixed-criticality systems as it allows for greater temporal and
spatial isolation of tasks of different criticalities.}

We conclude this section with some factors to consider when
designing a kernel lock for a real-time kernel such as Quest and a
discussion of the system model. 
To improve the predictability of the kernel lock, it needs to have
awareness of the priority of the task attempting to acquire the lock, i.e.,
it needs to be aware of the period of the VCPU that the task is mapped to.
Note that this paper only considers one-to-one mapping between tasks and
VCPUs\footnote[1]{We henceforth use tasks and VCPUs
interchangeably}. However, since all tasks mapped to a VCPU will have the
same period and budget, i.e., that of the VCPU, this does not have any impact
on the working of a predictable kernel lock. 
In certain cases it might be beneficial to allow tasks waiting for a
shared resource to suspend as it allows potentially lower priority tasks to
make progress. 
\annotate{A spinlock is generally used when the worst-case time required to
acquire the lock is smaller than the time taken to suspend the current waiter
and switch to a different task.}
However, since the scheduler is itself in the critical section,
and hence is accessed only after acquiring the kernel lock, there is no benefit in allowing suspension of contenders.
Thus, a modified spinlock that has priority awareness is sufficient.

\begin{figure}[!htbp]
  \includegraphics[width=0.6\linewidth]{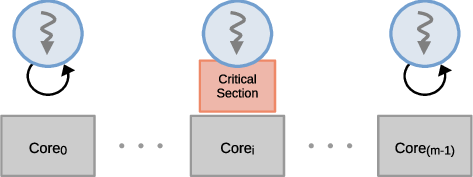}
  \caption{System Model}
  \label{fig:system_model}
  \Description{System Model: one task holding the lock on Core$_i$ and one waiter spinning on each of the $m-1$ other cores}
\end{figure}

{\bf System Model:} Given the aforementioned design requirements, we arrive at
the system model shown in Figure~\ref{fig:system_model}. We consider a
real-time system with $m$ homogeneous cores that uses spinlocks to guard
\annotate{the entire kernel. As the
scheduler itself is a part of the critical section, a} task waiting for, or holding a spinlock
on one of the cores cannot be preempted to allow another task to execute in its
place. Each task is characterized by a priority, $P$, and the index $i$ of its
current processor core. Task priorities are represented by unsigned integers in
the range $[P_{max}, P_{min}]$, with a lower number indicating a higher
priority level. The core index is a value in the range $[0,m)$. The system is
free to use static or dynamic priorities, under the condition that task
priorities are not modified while waiting to acquire a lock, or while executing
a critical section. The system requires a consistent notion of priority, with
global or semi-partitioned scheduling being allowed, given that task migrations
are permitted across lock invocations.

\section{Batched Priority Locking (BPL)}\label{sect:lock}
The Batched Priority Lock (BPL), shown in Algorithms~\ref{alg:BPL_lock}, and~\ref{alg:BPL_unlock},
aims to ensure bounded waiting time, while distinguishing between task
priorities among the same group of contenders. The following is a high level
overview of the algorithm:
\begin{itemize}
\item[1.] {\em Fast Path}: A task first checks if there are any other waiters,
and if there are none, it attempts to acquire the lock. This fast path enables
fast uncontested acquisition. If there are other waiters, or if the lock
acquisition fails, the task proceeds to the next step.
\item[2.] {\em Batching Stage} ($stage\_0$): All waiters in this stage obtain a
batch ID. Waiter(s) with the lowest batch ID, representing those that arrived
the earliest, are allowed to proceed to the next stage. Other waiters continue
to spin in the same stage. 
\item[3.] {\em Priority Ordering Stage} ($stage\_1$): Tasks in this stage
contend to determine the highest priority level (lowest $P$ value). Once they
settle on an agreed lowest $P$ value, all tasks of that priority proceed to the
next stage.
\item[4.] {\em Mutual Exclusion Stage}: In this final stage, the highest
priority task(s) among those that arrived in the earliest batch contend to
acquire the lock. 
\end{itemize}

The algorithm takes as inputs a pointer to a lock object ($L$), the
priority of a given task ($P$), and the index ($i$) of its corresponding
processor core.

\subsection{Attributes of the BPL Object}
\begin{itemize}
  \item $num\_waiters$ records the number of tasks contending for the
  lock and is initialized to 0. This value is tracked to choose the
  fast path of the lock in the absence of any contention. 

  \item $batch\_barrier$ and $priority\_barrier$ are key variables used by the
  tasks in $stage\_0$ and $stage\_1$. $batch\_barrier$ is set to the lowest batch ID
  of all waiters in $stage\_0$, while $priority\_barrier$ holds the highest priority
level of all waiters in $stage\_1$. Both barriers are initialized to
-1u, i.e.,
  the highest unsigned integer value. 

  \item $settling$ is an array holding two bitvectors, used to check if a task
  on a specific core is contending in a particular stage. For instance, if
  the $i^{th}$ bit of $settling[1]$ is set, it means there is a
  task on core $i$ that is contending in $stage\_1$. Initially, both
  the settling bitvectors are 0. 

  \item $curr\_batch$ is initialized to 0 and is split into two bit-fields: the
  lower $k$ bits track the number of tasks in a batch, and the remaining bits
  denote the batch ID for tasks entering $stage\_0$. A batch lasts for the
  duration of one critical section. Therefore, before a holder releases the 
  lock, it clears the lower $k$ bits and adds $2^k$ to $curr\_batch$, to start 
  a new batch ID. $curr\_batch$ is reset to zero only when there is no 
  contention for the lock. After reading $curr\_batch$ into a temporary 
  variable, new waiters shift their copy right by $k$ bits to get their batch
  ID. Note that the maximum number of waiters in a single batch is limited by
  $\min{(m-1,2^{k}-1)}$, where $k$ is set to ${\lceil}log_{2}{(m)}{\rceil}$ for
  an $m$-core system. 

  \item $status$ is a single bit that is set when the lock is held, and is 0
  when the lock is free. Initially the lock status is free.
\end{itemize}

\subsection{Explanation of the BPL Algorithm}

Consider the example shown in Figure~\ref{fig:bpl_working}:

\begin{figure}[!htb]
  \centering
    \includegraphics[width=0.81\linewidth]{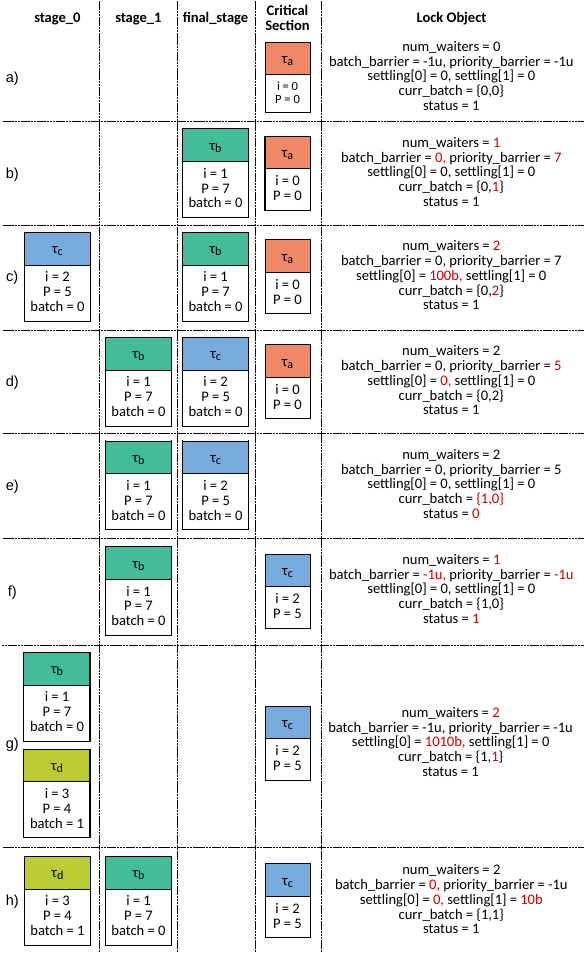}
	\caption{Batched Priority Locking (BPL) Algorithm Example}\label{fig:bpl_working}
  \Description{Example of the Batched Priority Locking (BPL) Algorithm}
\end{figure}

\begin{itemize}
  \item[a)] The lock is held by $\tau_a$, and the parameters of the lock object 
  are in their initial state as $\tau_a$ acquired the lock through the fast 
  path, which is explained in Subsection~\ref{subsect:bpl_impl}. 

  \item[b)] $\tau_b$ arrives next to contend for the lock. As the lock 
  is already held, it cannot execute the fast path. Hence, it increments
  $num\_waiters$, reads $curr\_batch$ to set its local batch variable, 
  and increments the lower field of $curr\_batch$ to 1 to indicate that there 
  is one task in batch 0. Since there are no other tasks in its batch, it is 
  able to set the $batch\_barrier$ and $priority\_barrier$ to its own batch and
  priority respectively, proceeding to the final spinlock stage. Note that it 
  sets bit 1 in the $settling[0]$ and $settling[1]$ bitvectors when it enters 
  $stage\_0$ and $stage\_1$, respectively. However, it resets them when leaving to 
  the next stage and, hence, they are both currently 0.

  \item[c)] Task $\tau_c$ next attempts to acquire the lock. It also increments 
  $num\_waiters$ and the lower field of $curr\_batch$, and then sets bit 2 of 
  $settling[0]$ before entering $stage\_0$. 

  \item[d)] Since $\tau_c$ holds the same batch ID as the $batch\_barrier$ (0),
  it proceeds to $stage\_1$, after resetting $settling[0]$. As its priority is
  higher than the $priority\_barrier$ (recall that lower $P$ value implies higher
  priority), $\tau_c$ updates $priority\_barrier$ to its own priority and proceeds to the final spinlock stage. Meanwhile, while $\tau_b$
  is spinning in the $final\_stage$, it sees that the $priority\_barrier$ no longer holds
  its own priority level, so it then moves back to $stage\_1$. Since both tasks are
  settled in the right stages, the $settling$ bitvectors are reset to 0 at
  this point. Thus, BPL ensures that within the same batch, higher priority 
  tasks acquire the lock first. 

  \item[e)] When $\tau_a$ completes its critical section, it increments the 
  batch ID field of $curr\_batch$ and clears the lower field, before resetting 
  the lock status to 0.

  \item[f)] As the lock is free, $\tau_c$ successfully sets the status bit, decrements $num\_waiters$, and enters its critical section. 
  To allow other tasks in its batch to move on to the final stage, 
  the new holder resets the $priority\_barrier$ before entering its critical section. 
  It could also be the last task in its batch, necessitating
  a reset of the $batch\_barrier$. Thus, whenever a task acquires 
  the lock, it resets both the barriers, forcing all the current waiters
  to reorganize themselves in the stages of the lock. 

  \item[g)] As both the barriers are reset, $\tau_b$ is pushed back to 
  $stage\_0$. At the same time, a new waiter $\tau_d$ arrives to contend for 
  the lock. Hence, $num\_waiters$ is incremented to 2. However, note that the 
  lower field of $curr\_batch$ now only holds the number of waiters in batch 1.
  Both tasks set their corresponding bits in $settling[0]$ and enter 
  $stage\_0$.

  \item[h)] Since $\tau_b$ belongs to an earlier batch than $\tau_d$, it 
  proceeds on to $stage\_1$, while $\tau_d$ stays in $stage\_0$. As it is
  settled in its stage, $settling[0]$ is 0, and only bit 1 is set in 
  $stage\_1$. Though not shown here, $\tau_b$ will then proceed to the final
  spinlock stage and acquire the lock when $\tau_c$ releases it. Thus, BPL 
  enforces FIFO ordering across batches.   
\end{itemize}

\subsection{Implementation of the BPL Algorithm}\label{subsect:bpl_impl}
We now discuss the implementation of the BPL algorithm in detail. In addition to common hardware instructions such as
the atomic \emph{increment (INC)}, \emph{decrement (DEC)},
\emph{store (STORE)},
\emph{set-bit (SET(variable, bit index))} and \emph{reset-bit (RESET(variable,
bit index))}, BPL requires the following atomic transactions:
\begin{itemize}
  \item \emph{test-and-set (TAS(bit))}: The TAS instruction returns 
  the value of $bit$ before setting it.
  \item \emph{compare-and-swap (CAS(variable, old value, new value))}: The CAS 
  instruction compares $variable$ with $old\ value$, and if they are equal, 
  $new\ value$ is stored in $variable$ and it returns $True$. 
  Otherwise, it returns $False$.
  \item \emph{fetch-and-add (FAA(variable, addend))}: The FAA 
  instruction returns the value of $variable$ before adding $addend$ to it.
\end{itemize}

The following is an explanation of Algorithms~\ref{alg:BPL_lock} and~\ref{alg:BPL_unlock} from the
perspective of task $\tau$ that attempts to first acquire and then eventually 
release a lock. 

\textbf{Fast Path (Lines~\ref{line:fast_path_begin}-\ref{line:get_batch_num})}:
If there are no other waiters, $\tau$ attempts to clear the $curr\_batch$ to 0
using an atomic CAS instruction. If there are any waiters when the batch ID is
reset, they could incur significant delays, because lower batch numbers are
first used to decide lock precedence. Hence, $num\_waiters$ has to be read
before attempting to clear $curr\_batch$. Additionally, if $curr\_batch$ only denoted the
value of the batch ID, it would only be updated when a lock holder releases the
lock. In that case, due to pathological interleaving, we could have a scenario
where one waiter clears $curr\_batch$ while another waiter has read a non-zero
batch ID. This could result in the latter waiting for an indefinite amount of
time to acquire the lock. To avoid such a scenario, we split $curr\_batch$ into
the batch ID and number of tasks fields. Thus, during the general flow of
control through the lock acquisition logic, all waiters first increment
$num\_waiters$ and then also modify $curr\_batch$. Any waiter that attempts to
clear $curr\_batch$ reads the two variables in opposite order: reading
$curr\_batch$ first in Line~\ref{line:fast_path_begin} and then $num\_waiters$
in Line~\ref{line:read_num_waiters}. This ensures that any changes made to
either variable are captured when a task attempts to clear it in
Line~\ref{line:reset_curr_batch}.

If the value of $curr\_batch$ changes (as a result of another task) between
Lines~\ref{line:fast_path_begin} and~\ref{line:reset_curr_batch}, the CAS
instruction fails. Irrespective of the success or failure of the CAS
instruction, $\tau$ tries to acquire the lock in
Line~\ref{line:fast_path_end} using an atomic TAS instruction on the
$status$. If this succeeds, $\tau$ jumps to the end of the \verb|lock| code
(Line~\ref{line:reset_barriers}). If this lock attempt fails, or if there are
other waiters, $\tau$ atomically increments $num\_waiters$
(Line~\ref{line:inc_num_waiters}), and then obtains a batch ID. It first
performs an FAA to atomically read and increment $curr\_batch$, and then
right-shifts the value $k$-bits to obtain the batch ID
(Line~\ref{line:get_batch_num}). It then proceeds to $stage\_0$ in the lock
acquisition.

\begin{algorithm}[!htbp]
  \caption{Batched Priority Lock - Lock Function}\label{alg:BPL_lock}
  \SetKwFunction{FLock}{lock}
  \SetKwProg{Fn}{function}{}{}
  {
  \fontsize{9.9}{11.869}\selectfont 
  \Fn{\FLock{struct bpl * L, uint32\_t P, uint32\_t i}}{
    
    prev = L$\rightarrow$curr\_batch;\label{line:fast_path_begin}\\
    \If {L$\rightarrow$num\_waiters == 0\label{line:read_num_waiters}} {
      CAS (L$\rightarrow$curr\_batch, prev, 0);\label{line:reset_curr_batch} \tcp{Fast Path}
      \lIf {!TAS (L$\rightarrow$status)} {goto \hyperref[line:reset_barriers]{acqd}}\label{line:fast_path_end}
    }
    
    INC (L$\rightarrow$num\_waiters);\label{line:inc_num_waiters} \\      
    batch = FAA (L$\rightarrow$curr\_batch, 1);\label{line:inc_curr_batch}\\
    batch = batch $>>$ k;\label{line:get_batch_num}
    
    \tcp{FIFO / Batching Stage}
    {\bf stage\_0}: SET (L$\rightarrow$settling[0], i);\label{line:stage0_begin}

    {\bf read\_batch\_barrier}: prev = L$\rightarrow$batch\_barrier; \label{line:read_barrier_0} \\
    \If {batch $\leq$ prev\label{line:stage0_batch_check} } {
      \If {CAS (L$\rightarrow$batch\_barrier, prev, batch)} {
        RESET (L$\rightarrow$settling[0], i); \label{line:lowest_batch}
        \tcp{control goes to Line~\ref{line:stage0_wait_settle}}
      }
      \lElse {goto \hyperref[line:read_barrier_0]{read\_batch\_barrier}}
    } \Else {
      RESET (L$\rightarrow$settling[0], i);
      goto \hyperref[line:read_barrier_0]{read\_batch\_barrier};\label{line:not_lowest_batch}
    }

    \lWhile {L$\rightarrow$settling[0] != 0} {}\label{line:stage0_wait_settle}
    \lIf {L$\rightarrow$batch\_barrier != batch} {goto \hyperref[line:stage0_begin]{stage\_0}}\label{line:barrier_0}
    
    \tcp{Priority Ordering Stage}
    {\bf stage\_1}: SET (L$\rightarrow$settling[1], i);\label{line:stage1_begin} \\
    
    {\bf read\_priority\_barrier}: prev = L$\rightarrow$priority\_barrier;\label{line:read_barrier_1} \\

    \If {L$\rightarrow$batch\_barrier != batch\label{line:stage1_batch_check}} {
      STORE (L$\rightarrow$priority\_barrier, -1u); \\
      RESET (L$\rightarrow$settling[1], i);
      goto \hyperref[line:stage0_begin]{stage\_0};
    }

    \If {P $\leq$ prev\label{line:priority_check}} {
      \If {CAS (L$\rightarrow$priority\_barrier, prev, P)}{
        RESET (L$\rightarrow$settling[1], i);\label{line:highest_priority}
        \tcp{control goes to Line~\ref{line:stage1_wait_settle}}
      }
      \lElse {goto \hyperref[line:read_barrier_1]{read\_priority\_barrier}} 
    } \Else {
      RESET (L$\rightarrow$settling[1], i);
      goto \hyperref[line:read_barrier_1]{read\_priority\_barrier};\label{line:not_highest_priority}
    }

    \lWhile {L$\rightarrow$settling[1] != 0} {}\label{line:stage1_wait_settle}
    
    {\bf final\_stage}:\label{line:finalstage_begin} \tcp{Spinlock} 
    
    \lIf {L$\rightarrow$priority\_barrier != P\label{line:final_P_check}} {goto
    \hyperref[line:stage1_begin]{stage\_1}} 

    \If {L$\rightarrow$batch\_barrier != batch\label{line:final_batch_check}} {
      STORE (L$\rightarrow$priority\_barrier, -1u);
      goto \hyperref[line:stage0_begin]{stage\_0};
    }

    \lIf {!TAS (L$\rightarrow$status)\label{line:final_TAS}} {goto
    Line~\ref{line:dec_counters}}
    \lElse {goto \hyperref[line:finalstage_begin]{final\_stage}}
    
    \tcp{Lock Acquired}
    DEC (L$\rightarrow$num\_waiters);\label{line:dec_counters} \\
    
    {\bf acqd}:\label{line:reset_barriers}
    STORE (L$\rightarrow$priority\_barrier, -1u);
    STORE (L$\rightarrow$batch\_barrier, -1u);
  }
  }
\end{algorithm}

\begin{algorithm}
  \caption{Batched Priority Lock - Unlock Function}\label{alg:BPL_unlock}
  \SetKwFunction{FUnlock}{unlock}
  \SetKwProg{Fn}{function}{}{}
  \Fn{\FUnlock{struct bpl * L}}{      
    new\_val = L$\rightarrow$curr\_batch; \\
    AND(new\_val, (($2^{32}$ - 1) - ($2^k$ - 1))); \\
    ADD(new\_val, $2^k$); \\
    STORE (L$\rightarrow$curr\_batch, new\_val);\label{line:inc_batch} \\
    RESET (L$\rightarrow$status, 0);\label{line:unlock}
  }
\end{algorithm}

\textbf{stage\_0 (Lines~\ref{line:stage0_begin}-\ref{line:barrier_0})}: When
$\tau$ enters $stage\_0$, it sets the bit corresponding to its core
($i$) in the $settling[0]$ bitvector. This bitvector is used to ensure that
a task that is the first to reach the end of $stage\_0$ proceeds
to $stage\_1$ only when there are no other waiters, or all waiters have
settled on a specific value as the lowest batch ID. $\tau$ then reads
$batch\_barrier$, which holds the ID of the earliest batch. If $\tau$
belongs to that batch or an earlier batch
(Line~\ref{line:stage0_batch_check}), it performs an atomic CAS
transaction. Note that this branch can be further split into tasks
in earlier batches ($< batch\_barrier$) and those in the current batch ($=
batch\_barrier$), and allowing only the tasks in earlier batches to perform the
CAS. However, by not splitting them we keep the implementation of the algorithm
in x86 assembly simpler.
If the CAS transaction succeeds, $\tau$ resets its settling bit 
(Line~\ref{line:lowest_batch}), else it tries again. If $\tau$ finds
that it belongs to a later batch, it resets its settling bit, and
continues to check $batch\_barrier$
(Line~\ref{line:not_lowest_batch}). Thus, all waiters will reset their
settling bit either after successfully setting $batch\_barrier$ to their
batch in Line~\ref{line:lowest_batch}, or after finding that they do
not fall into the earliest batch in
Line~\ref{line:not_lowest_batch}. Tasks that succeed in the CAS
operation check the $settling[0]$ bitvector to ensure that all other
waiters have reset their settling bits, i.e., compared themselves with
$batch\_barrier$ before proceeding to the next stage
(Line~\ref{line:stage0_wait_settle}).  Line~\ref{line:barrier_0} acts
as the final step in the barrier, pushing back any tasks that do not
have the lowest batch ID but just happened to finish the transaction
first.

\textbf{stage\_1 (Lines~\ref{line:stage1_begin}-\ref{line:stage1_wait_settle})}:
When $\tau$ enters $stage\_1$ in Line~\ref{line:stage1_begin}, similar to
$stage\_0$, it first indicates its presence by setting the $settling[1]$ bitvector.
As a next step, in addition to reading the value of $priority\_barrier$, which holds
the highest priority level of all tasks in $stage\_1$, $\tau$ also checks if it
still is the lowest batch task (Line~\ref{line:stage1_batch_check}). If not,
$\tau$ clears its settling bit and resets $priority\_barrier$ to -1u (unsigned) before jumping back
to $stage\_0$. The following scenario motivates the need for this check.
Consider multiple tasks of the same highest priority spinning in the final
stage. One of them acquires the lock, hence it first resets
$priority\_barrier$ before also resetting $batch\_barrier$. One or more of the tasks in the $final\_stage$ could be at
Line~\ref{line:final_P_check}. As soon as $priority\_barrier$ is
reset, these tasks will move back to $stage\_1$. In
$stage\_1$, they first read the current $priority\_barrier$ value
(Line~\ref{line:read_barrier_1}). Now, if the check in 
Line~\ref{line:stage1_batch_check} were not present, these tasks would then
immediately try to re-enter the $final\_stage$ by performing a CAS on $priority\_barrier$.
However, there could be other tasks in $final\_stage$ that got pushed all the
way back to $stage\_0$ because they were in Line~\ref{line:final_batch_check}
when
the barriers were reset. To allow these tasks to get equal footing, we perform
this additional check. Also, this check protects against the condition where a very
fast waiter of a newer batch acquires the lock before a slow waiter that belongs
to an older batch. In the next \verb|if-else| block
(Lines~\ref{line:priority_check}-\ref{line:not_highest_priority}), similar to
$stage\_0$, the waiting task performs a CAS on $priority\_barrier$ if its $P$ value is
the smallest in magnitude among all tasks in $stage\_1$. Tasks that succeed in
this CAS then wait until all tasks reset their $settling[1]$ bit in
Line~\ref{line:stage1_wait_settle}.

\textbf{final\_stage
(Lines~\ref{line:final_P_check}-\ref{line:reset_barriers})}: Only tasks that
have the lowest batch value ($batch$) among all contenders and the highest
priority (lowest $P$ value) among those of the same batch should be allowed to
perform the final bit TAS instruction. Hence, when $\tau$ enters the final
stage, before another check of the lock status in
Line~\ref{line:final_TAS}, it also checks $priority\_barrier$
(Line~\ref{line:final_P_check}) and $batch\_barrier$
(Line~\ref{line:final_batch_check}). 
When the lock is freed, the task in the $final\_stage$ succeeds in the
TAS operation. It then decrements the number of waiters, and resets
the barriers, allowing the other waiters to contend to determine the
next holder.

\textbf{Unlock (Algorithm~\ref{alg:BPL_unlock})}: By including all the ordering logic in the acquiring side of the lock,
we simplify the releasing task's requirements. The {\tt unlock}
operation, shown in Algorithm~\ref{alg:BPL_unlock} simply updates the current batch ID and then resets the
status bit (Line~\ref{line:unlock}).  The batch ID is updated by first
reading it into a local variable, resetting the lower $k$ bits,
incrementing the upper $32 - k$ bits by adding $2^k$, and storing it in the lock object, using the atomic STORE instruction. This resets
the count of the number of waiters while incrementing the batch ID by
1, assuming a 32-bit word-size. A CAS is unnecessary here because the
number of waiters cannot be more than $m-1$ and
$k={\lceil}log_2(m){\rceil}$.

\subsection{Discussion}
Irrespective of how many times a waiter is pushed back to $stage\_0$ or $stage\_1$, the
batch ID is always the first parameter used to determine the next waiter. Each
waiter reads the batch ID only once (in Line~\ref{line:inc_curr_batch}). When
there is contention for the lock, i.e., $num\_waiters$ is non-zero, the batch ID
is monotonically non-decreasing. When there is no contention, $num\_waiters
= 0$, and the waiter will immediately acquire the lock. Thus, in both cases, the
lock is acquired in a finite number of steps. The settling
bitvector in each stage ensures that a fast process that completes a stage early
is forced to wait for all other tasks in that stage to also settle. This ensures
that the batch order is followed in practice, providing the worst-case bound of
FIFO locks. 

While a FIFO lock considers only the order of lock request, and a
priority-ordered lock only considers the priority of waiters, BPL
considers both. It reaches a middle ground, minimizing priority
inversions in the average case while still meeting the worst-case
bound of FIFO. In the next section, we analyze the properties of BPL.

\section{Analysis of BPL}\label{sect:analysis}
This section analyzes the key properties of BPL, including the
asymptotic execution costs of an implementation on an \emph{m-core
system}. 

\subsection{Progress Guarantees}

\begin{theorem}\label{th:starvation}
  A Batched Priority Lock is starvation-free.
\end{theorem}

\begin{proof}
  Even though it considers priority among waiters in the same batch,
  BPL first groups tasks by the window in which they made their lock
  requests. The previous batch is closed and a new batch is created in
  Line~\ref{line:inc_batch} of Algorithm~\ref{alg:BPL_lock}, when a holder
  increments the batch number just before releasing the lock. Thus,
  BPL ensures that if a task makes a lock attempt at time $t$, it
  will not have to wait for any task that makes a lock attempt after
  $t+CS$, where $CS$ is the duration of the critical section of the
  task holding the lock at time $t$. In other words, all tasks that
  arrive within the same critical section (while the lock is held)
  will acquire the lock before any tasks that arrive in the next
  critical section. Therefore, even in cases of high contention, BPL
  guarantees progress for all waiters.
\end{proof}
 
\begin{lemma}\label{lem:max-batch-size}
  The maximum batch size for a Batched Priority Lock in
  Algorithms~\ref{alg:BPL_lock}, and~\ref{alg:BPL_unlock} is $m-1$.
\end{lemma}

\begin{proof}
  In an $m$-core system that uses spin-based locks to guard
  non-preemptible critical sections, no more than $m$ tasks can
  contend to acquire a free lock. This is because neither the
  waiting tasks nor the lock holder are preemptible. One of the $m$
  contenders will succeed in acquiring the lock through the fast
  path in Line~\ref{line:fast_path_end}, and hence will not read the
  batch ID. In the worst-case, up to $m-1$ tasks will fetch and
  increment $curr\_batch$ while waiting for the lock. With
  $m-1$ waiters and $1$ holder, a new contender can only arrive
  after the holder releases the lock. By updating $curr\_batch$ and
  starting a new batch \emph{before releasing the lock}, BPL ensures
  that a batch grows no larger than $m-1$.
\end{proof}

\begin{lemma}\label{lem:sum-of-batches}
  The total number of tasks across all concurrent batches in the Batched
  Priority Lock is bounded by $m-1$.
\end{lemma}
\begin{proof}
  Since the maximum number of tasks involved in a single lock cannot exceed
  $m$, irrespective of the number of concurrent batches, the total number of
  waiters in all of them is also bounded by $m-1$.
\end{proof}

\begin{theorem}\label{th:worst-case}
  In the worst-case, a task using a Batched Priority Lock will wait
  for no more than $m-1$ critical sections before entering its
  critical section.
\end{theorem}

\begin{proof}
In the worst-case, a task $\tau$ attempting to acquire a lock will have to
wait for (1) the current holder, (2) tasks in all earlier batches, and (3)
all higher priority tasks in its own batch. 
In Case 1, no waiter can be attempting to acquire the lock on the same core as
the current holder, since critical sections are non-preemptible and locks are
non-blocking. 
Cases 2 and 3 represent all other waiters excluding $\tau$.
Since waiters busy-wait and do not block, they must each be on a distinct core.
Hence the total number of tasks in Cases 2 and 3 cannot be greater than $(m-2)$.
Therefore the maximum number of critical
sections that $\tau$ will have to wait to acquire the lock is $1$, from Case 1,
and $(m-2)$, from Cases 2 and 3. Thus, the total cannot be greater than $(m-1)$.
\end{proof}

\subsection{Blocking Delay and Schedulability}

BPL guarantees worst-case FIFO-bounded delay as long as a task in a
later batch is not able to gain the lock before a task in an earlier
batch. This is assured for homogeneous cores, which are assumed to
execute stages of the acquisition function at the same rate, or for
cores where the disparity in speeds is relatively low. Violations to
the FIFO bound would only be possible if a task in a later batch on
one core could pass through all stages of the lock acquisition
algorithm before any other task in an earlier batch completed Line~\ref{line:stage0_begin}
of Algorithm~\ref{alg:BPL_lock}. 

While the worst-case blocking bound remains identical for FIFO locks
and BPL, implying equivalent schedulability analysis in theory, BPL provides
superior average-case performance for higher priority tasks within each batch.
Under FIFO locks, all tasks experience approximately the same average waiting
time, irrespective of priority. In contrast, BPL reduces the average waiting
time for higher-priority tasks while increasing it for lower-priority tasks.
The cost of implementing BPL is slightly higher than that of a FIFO-ordered
lock (as discussed in Section~\ref{sect:uncontested_acquisition}). However, if
the reduced delay to higher priority tasks outweighs the implementation costs,
then BPL should be seen as preferred to a FIFO lock.

\subsection{Implementation} 
Variants of the instructions used by BPL 
exist for most popular architectures making the lock approach portable. For
example, the x86 architecture features CMPXCHG, XADD, ADD, INC, DEC, BTS, XCHG,
MOV,
AND, OR, and BTR instructions, some of which may have a LOCK prefix to enforce
atomicity. On ARM, the LDREX/STREX instructions along with explicit
memory barriers (DMB) can be used to implement BPL.

Since shared memory is modified using the CAS instruction in a loop in each
stage, in the worst-case it takes a waiter $O(m)$ time to settle in an
appropriate stage in the \verb|lock()| function. The
\verb|unlock()| function does not contain any loops and hence is executed in
constant time. Apart from the lock object, which is globally accessible and
hence must be fetched from memory, all other variables are local and stored in
registers to minimize memory access overhead. 

The ticket lock shown in Algorithm~\ref{alg:ticket_lock} is a simple
FIFO-ordered lock. It uses two counters: one incremented at request
time, and the other at release time. This algorithm enforces each task
to be ordered by a unique request value, or a batch number, which is
acquired when attempting to access the lock.

\begin{algorithm}[!htbp]
  \caption{Ticket Lock}\label{alg:ticket_lock}
  \SetKwFunction{FLock}{lock}
  \SetKwFunction{FUnlock}{unlock}
  \SetKwProg{Fn}{function}{}{}
  \Fn{\FLock{struct fl * L}}{

    batch = FAA (L$\rightarrow$request, 1);\label{line:fl_request}\\

    check\_response: \If {L$\rightarrow$release != batch} {
      goto check\_response;
    }
  }
  \Fn{\FUnlock{struct fl * L}}{      
    INC (L$\rightarrow$release);\label{line:fl_release}    
  }
\end{algorithm}

BPL similarly uses a FIFO ordering among batches in $stage\_0$ of
Algorithm~\ref{alg:BPL_lock}, where more than one task has the same batch number. The
ticket lock avoids pathological interleaving of waiters and holders by using two
counters: one incremented by the waiters and the other by the holder when
releasing the lock. These counters are equivalent to $num\_waiters$ and
$curr\_batch$ in BPL. 

Assuming a 32-bit architecture, the batch ID is denoted by the upper
$32 - k$ bits of $curr\_batch$. The largest value that can be
represented by an unsigned integer variable of size $32 - k$ is $2^{32
- k} - 1$. If the batch ID is incremented beyond this value, it will
lead to the variable looping back to 0, which may lead to an
arbitrarily long wait time for tasks that hold a higher batch ID. To
avoid this issue, Line~\ref{line:reset_curr_batch} resets the current
batch ID to 0, after ensuring there are no other waiters and the
current value is the same as the value read in
Line~\ref{line:fast_path_begin}. Thus, BPL requires there to be a
brief period of no contention for every $2^{32 - k} - 1$ lock
acquisition attempts. For example, if there are 64 cores in the
system, and $k$ is set to $6$, BPL requires there to be a small period
of no contention for every $2^{26}-1$ (more than 67 million) lock
requests. At this point, all current waiters should complete their
critical sections, before new lock requests are made. In a 64-bit
system, this requirement is even less restrictive, allowing more than
$2.8{\times}10^{17}$ requests on a 64-core machine before requiring a
no-contention batch ID reset. Thus, we do not consider it to be a
practical limitation of the implementation.

If multiple tasks in the earliest batch have the same highest
priority, all of them will reach the final stage
(Line~\ref{line:final_P_check}). The TAS lock in
Line~\ref{line:final_TAS} ensures that only one of them succeeds in
acquiring the lock. While this choice is made arbitrarily, the
worst-case lock delay bound among tasks within the batch remains the
same as derived in the previous subsection.

Cancellation of a request is not typically required in the case of
kernel locks. However, BPL is able to support lock cancellations with minimal
changes. Any task that intends to terminate its request simply resets
$batch\_barrier$ and $priority\_barrier$ before exiting the \verb|lock()|
function, triggering a reordering of the remaining waiters.

\section{Experimental Evaluation}\label{sect:exp_eval}
\subsection{Simulation}\label{sect:sim-BPL}

We use a queuing theory approach to model task requests for a batched priority
lock. As will be seen in the evaluations, BPL is compared against FIFO and
priority-based locks under simulated conditions.

We begin with the machine repairman queuing theory model~\cite{queuing_theory}, shown in
Figure~\ref{fig:mm1kk}. The model consists of $m$ sources of traffic, which
require service from a single server. The server handles only one request at a
time. Since a request represents a machine requiring repair, a source cannot
generate another request until its previous one has been serviced.
In our case, due to the non-preemptible nature of waiters and lock holders, the
sources are tasks mapped to $m$ cores, and the server is a
shared resource (e.g., kernel) guarded by a lock. Most importantly, similar to
the sources in the model, tasks running on separate cores cannot request the
lock to the shared resource until their previous attempt has completed.

\begin{figure}[!htb]
  \centering
    \includegraphics[width=0.38\linewidth]{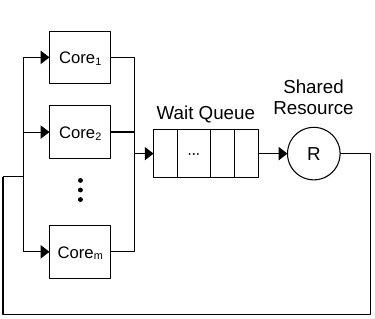}
	\caption{Machine Repairman Queue Model with $m$ Sources}\label{fig:mm1kk}
  \Description{Machine Repairman Queue Model with m Sources}
\end{figure}

One approach to modeling request arrivals from each of the $m$ sources is to
consider them as separate Poisson arrival processes. Each source, $C_i$, has an
arrival rate, $\lambda_i$, with an expected exponentially distributed
inter-arrival time, $\frac{1}{\lambda_i}$. For $m$ separate Poisson arrivals,
one per core, the model has an aggregated arrival rate of $\lambda_{agg} =
\sum\limits_{i = 1}^{m}{\lambda_i}$. If the shared resource is managed by a
single server also having a Poisson service rate, $\mu$, then its expected
service time is $\frac{1}{\mu}$.

Rather than have each core generate independent arrivals, we use a traffic
generator that produces concurrent bursts of arrivals. A burst of size $m$ implies $m$ cores coordinating to
generate one lock request each at the same time instant. The generator has an average rate
$\lambda_{burst}$. Therefore, the inter-arrival time between consecutive
executions of the burst generator is exponentially distributed with a mean value
of $\frac{1}{\lambda_{burst}}$. Each time the generator runs, it produces a
burst of arrivals from a uniform distribution having a specific mean. As a
uniform distribution over a range $[0,m]$ has an expected value of $m/2$, it
never generates bursts with a mean value greater than this.
For compliance with our queuing model, the burst requests are randomly selected
from sources that do not already have a request pending. If the number of
available sources is less than the generated burst size, all of them are
selected. If there are no available sources, the burst generator waits for
another interval of time.
Thus, the number of sources picked might be less than or equal to the random
burst size.
The generator then triggers the chosen sources to each submit a lock request.
This queuing model is simulated using the SimPy discrete event
simulator~\cite{simpy}. The different configurations that are evaluated are as
follows:
\begin{itemize}
  \item {\em Number of Cores (Sources)}: 8, 16, 32, and 64
  sources all having distinct priorities.
  \item {\em Mean Burst Size}: chosen in powers of 2, from $2^2$ to $m/2$.
  \item {\em Service Rate ($\mu$)}: this is kept at 0.01, to represent an average
  of 1 request serviced every 100 time units.
  \item {\em Burst Arrival Rate ($\lambda_{burst}$)}: For each core count and mean burst
  size, $\lambda_{burst}$ is varied from $0.01\mu$ to $1.0\mu$.
  \item Lock Ordering: FIFO (FL) and strict priority ordering (PL) are simulated
  using the shared resource primitives available in SimPy. A custom class was
  added to model Batched Priority Lock ordering (BPL). Since our only goal here
  is to understand the impact of different locking methods, we do not include
  the fast path, which is an optimization provided for practical
  performance.
\end{itemize}

We use the \emph{weighted mean delay} to evaluate the efficacy of each lock
ordering. If the mean delay of $Core_i$ is denoted by $d_i$ and its weight is
denoted by $w_i$, the weighted mean delay, $d_w$ is obtained using the following
equation:
\begin{equation}
  d_w = \frac{\sum\limits_{i = 1}^{m}{w_i \cdot d_i}}{\sum\limits_{i = 1}^{m}{w_i}}
\end{equation}

In our tests, the lowest priority core is assigned a weight of 1, and each
higher priority core is assigned a weight one higher than the previous core.
Thus, the highest priority core has the weight $m$. Each test is run until a
total of $m{\times}10000$ lock requests are complete.
The results obtained are consistent across different core counts. Hence, we show only the results of the 64-core test case with mean
burst sizes 8 and 32. The total priority inversions, and weighted mean delay
(normalized to that obtained for FIFO ordering), are shown in
Figure~\ref{fig:sim_64cores}.

\begin{figure*}
  \centering
  \includegraphics[width=0.3\textwidth]{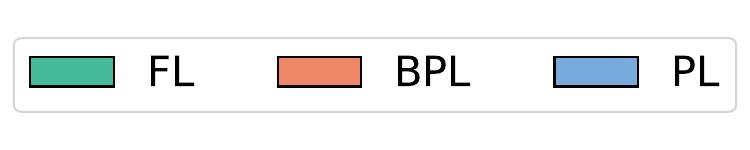}
  

  \begin{subfigure}[b]{0.49\linewidth}
    \centering
    \includegraphics[width=\linewidth]{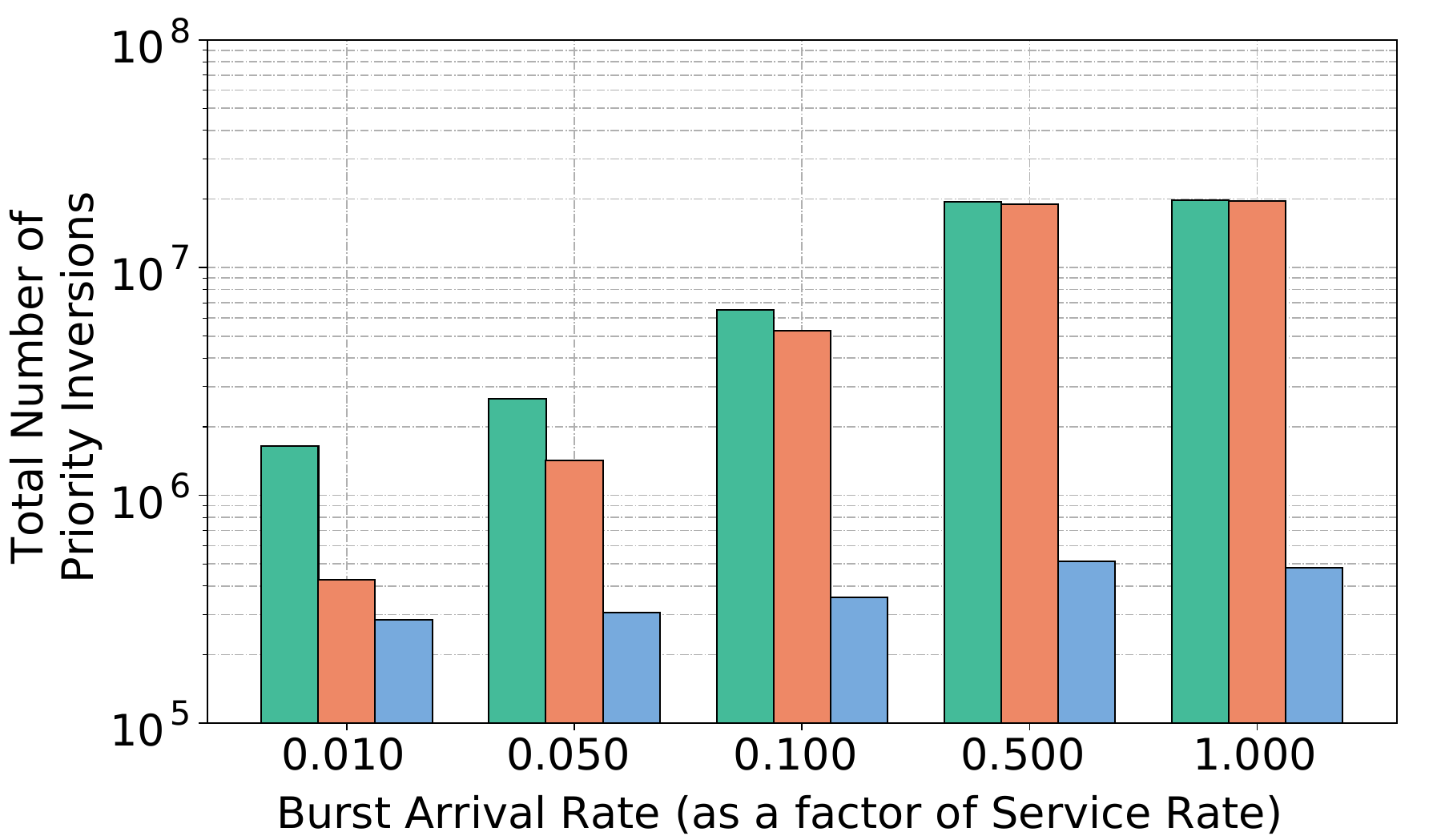}
    \caption{Priority Inversions (mean burst size = 8)}
    \label{fig:mbs8_pinv}
  \end{subfigure}
  \hfill
  \begin{subfigure}[b]{0.49\linewidth}
    \centering
    \includegraphics[width=\linewidth]{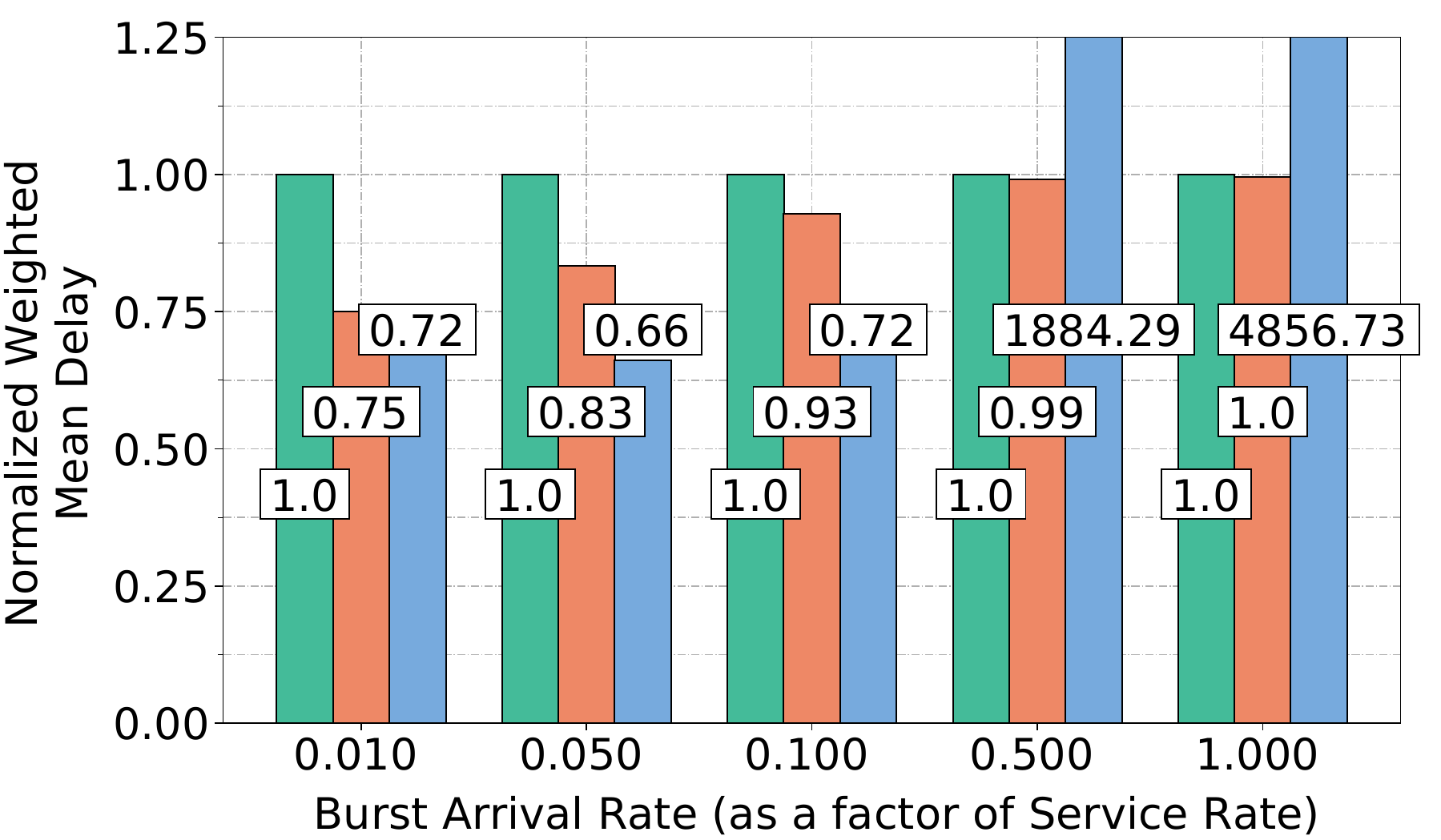}
    \caption{Weighted Mean Delay (mean burst size = 8)}
    \label{fig:mbs8_wavg}
  \end{subfigure}

  \vspace{0.3cm}
  
  \begin{subfigure}[b]{0.49\linewidth}
    \centering
    \includegraphics[width=\linewidth]{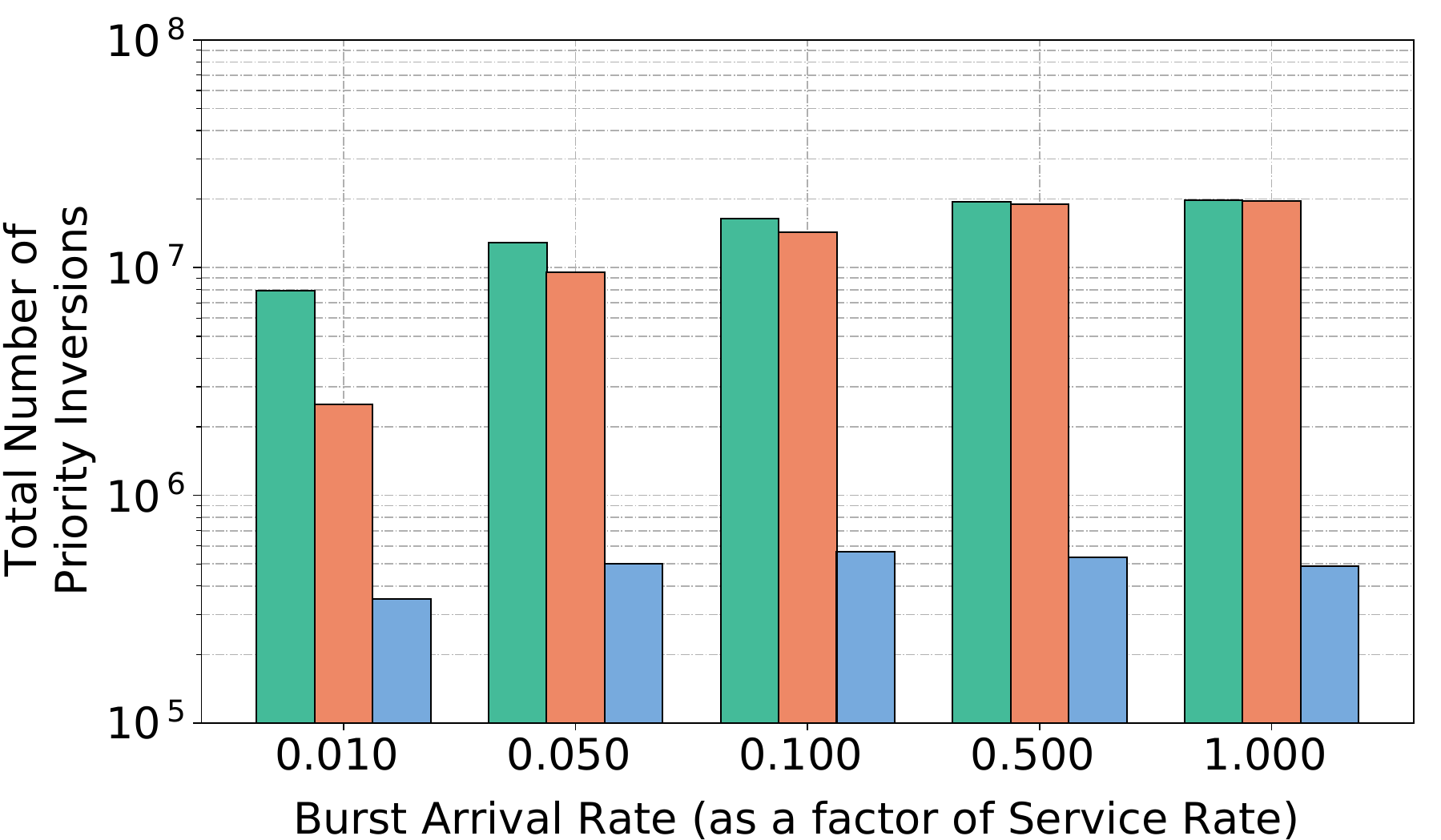}
    \caption{Priority Inversions (mean burst size = 32)}
    \label{fig:mbs32_pinv}
  \end{subfigure}
  \hfill
  \begin{subfigure}[b]{0.49\linewidth}
    \centering
    \includegraphics[width=\linewidth]{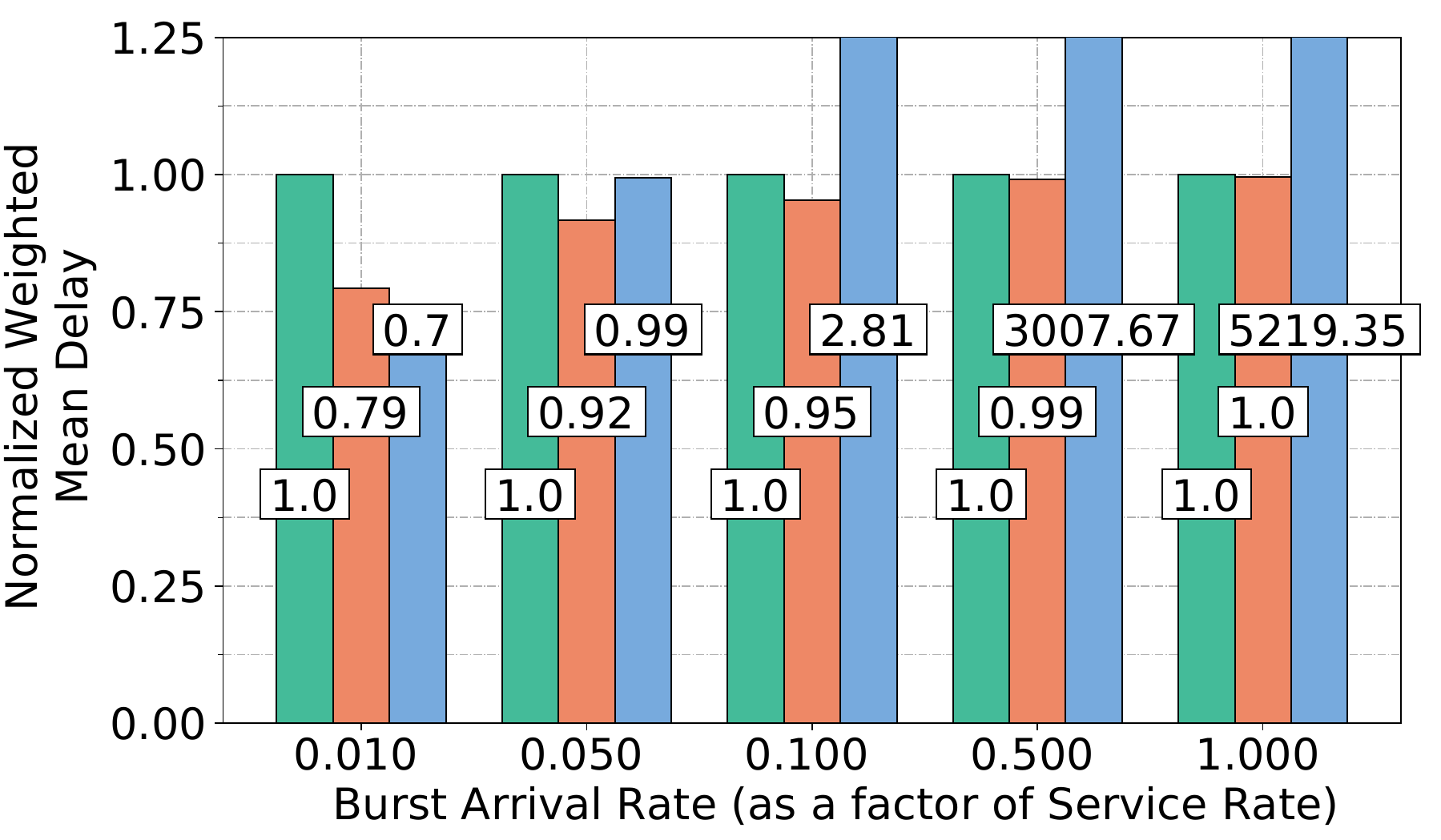}
    \caption{Weighted Mean Delay (mean burst size = 32)}
    \label{fig:mbs32_wavg}
  \end{subfigure}
  
  \caption{Comparison of Priority Inversions and Weighted Mean Delay for 64 Cores}\label{fig:sim_64cores}
  \Description{Comparison of Priority Inversions and Weighted Mean Delay for 64 Cores}
\end{figure*}

{\bf Observation 1:} For lower mean burst sizes and arrivals, the
performance of BPL is close to that of PL.

In Figures~\ref{fig:mbs8_pinv} and~\ref{fig:mbs32_pinv}, for lower mean burst sizes and arrival rates, the
percentage of lock requests experiencing priority inversions under Batched
Priority Locking (BPL) is significantly lower than for FIFO ordering (FL),
although never as good as for strict priority ordering (PL). 
This is because, when the arrival rate is low, the inter-arrival time is
long enough for most if not all requests generated in each burst to be
serviced. Hence, the behavior of BPL is very close to that of PL: all
requests in a burst are categorized into a batch, and they are allowed access
to the resource in priority order.

{\bf Observation 2:} For higher mean burst sizes and arrivals, the
performance of BPL is never worse than that of FL.

As the burst arrival rate increases, we see a divergence in the performance of
PL and BPL, and a convergence between FL and BPL. Higher arrival rates yield new
bursts while requests from one or more previous bursts are still waiting to be
serviced. This potentially splits waiters into increased numbers of separate
batches, if arrivals span the execution of different critical sections, so the
benefits of priority ordering within a given batch are diminished. For very high
arrival rates, each batch size approaches 1, and there can be up to $m-1$ such
batches, so BPL degrades to FL. 
In contrast, at lower burst arrival rates and
sizes, a group of tasks may join the same batch, and BPL performs closer to PL
in terms of minimizing priority inversions. 
Note that lower arrival rates do not necessarily imply low 
contention, as contention also depends on the burst size. Thus, for a mean 
burst size of 32, even when the arrival rate is only 0.010 times the service
rate, an average of 32 tasks contend for the lock at each arrival.

Figures~\ref{fig:mbs8_wavg} and~\ref{fig:mbs32_wavg} provide insights into the impact of priority
inversions on task waiting times. In cases where BPL avoids most priority
inversions and performs similarly to PL, the weighted mean delays of the two
orderings are comparable. The difference between FL, and the pair of prioritized
locks in these cases is proportional to the difference in the priority
inversions between them. As the burst arrival rate increases, the normalized
weighted mean delay of BPL approaches that of FL. However, PL now suffers badly,
because it starves lower priority tasks from making progress altogether. The
graphs are capped at a normalized $d_w$ value of 1.25, but PL reaches values
above 5,000 in some cases. This shows the limitations of a strict priority
ordering, which attempts to avoid priority inversions in all cases at the cost
of potentially unbounded delay for lower priority tasks.

These results show that Batched Priority Locking avoids priority inversions as
long as progress guarantees are not impacted. BPL generally outperforms FL in
terms of minimizing priority inversions and reducing waiting times for higher
priority tasks. When there are large bursts of tasks arriving
infrequently, they get grouped into the same batch, with BPL performance being
similar to that of PL. On the other hand, if batch sizes reduce to 1 task, even under high arrival rates, BPL avoids the
unbounded delays of a strictly priority-ordered lock, limiting the
worst-case delay to that of FL.

\subsection{Evaluation in an RTOS}\label{sect:rtos_eval}

In this subsection, we implement and evaluate BPL using Quest RTOS running on a
Cincoze DX1100 Embedded PC~\cite{DX}. This platform has a 1.8 GHz
Intel Core i7-9700TE CPU with 8 physical cores. BPL is compared to an
unordered spinlock (SL) implemented using the TAS lock and FIFO-ordered ticket
lock (FL). We omit comparisons with strictly priority-ordered lock as it has
already been shown to be unsuitable due to its potential to starve lower priority tasks.
We omit comparisons against MCS, and other
blocking locks that use queues, because BPL is designed for systems
that use global spin-based locks.

\subsubsection{Cost of Uncontested Lock Acquisition}\label{sect:uncontested_acquisition}


The overhead of each lock is evaluated by measuring the time taken to
acquire and release it using the \verb|rdtsc| instruction when there
are no contenders. 
All measurements are performed with interrupts
disabled, so as not to skew results and are averaged over 10,000
readings. After
excluding the 45 cycle overhead of the \verb|rdtsc| instruction
itself, the observed minimum, median, 99.9th percentile and maximum
overhead in clock cycles are shown in Figure~\ref{fig:uncontested_acquisition}.

\begin{figure}[ht]   
  \centering
    \includegraphics[width=0.5\linewidth]{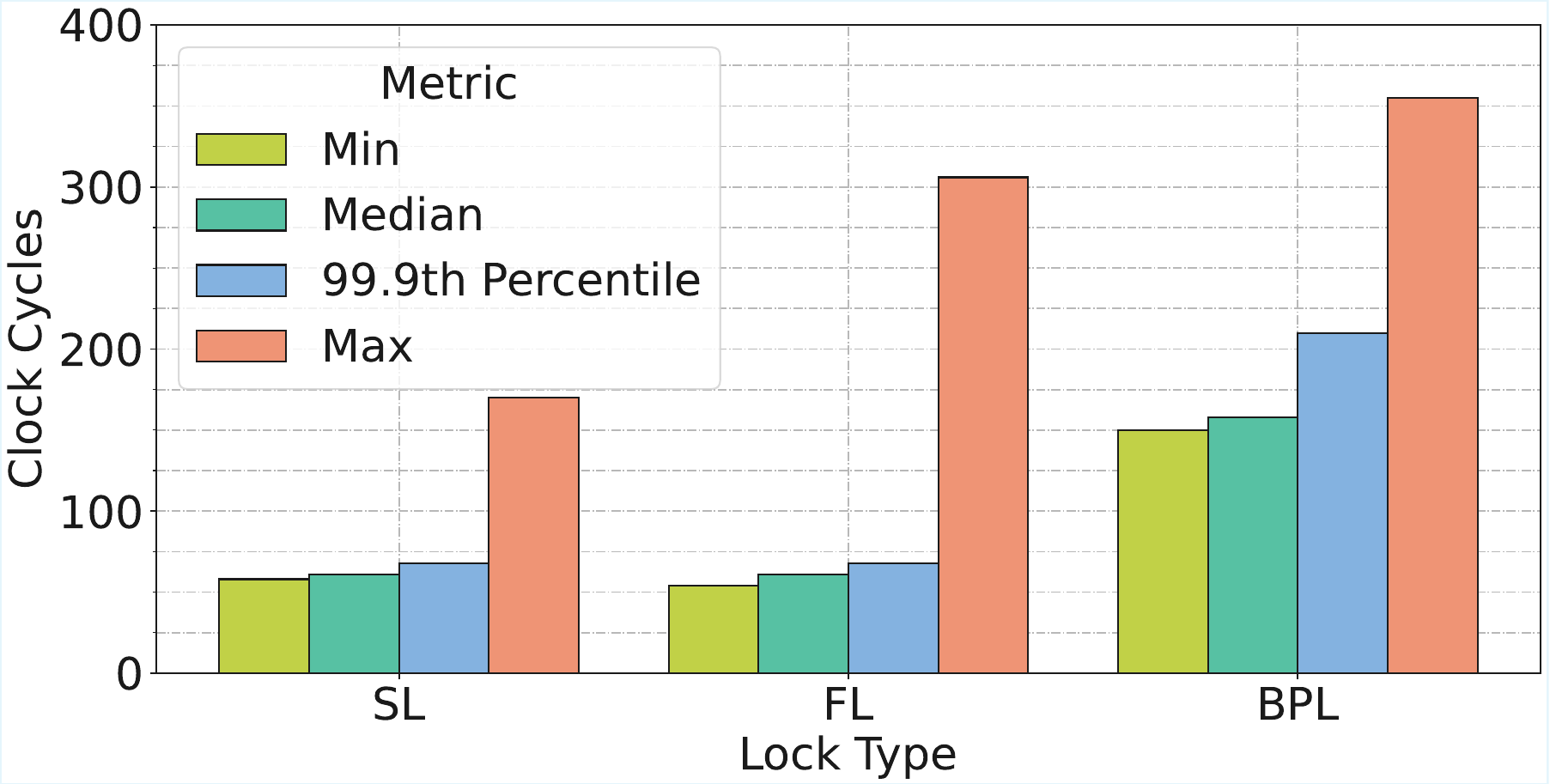} 
  \caption{Cost of Uncontested Lock Acquisition}\label{fig:uncontested_acquisition}
  \Description{Cost of Uncontested Lock Acquisition}
\end{figure}

For both in the minimum and median cases, BPL incurs $\approx100$ clock cycle
increase in overhead over both the simple spinlock and the ticket
lock. Even in the 99.9th percentile case, the increase in overhead
from the unordered spinlock to BPL is less than 150 cycles. This low
overhead shows the effectiveness of the fast path.  The maximum time
recorded for all the locks occurs at the very first instance of
locking and unlocking. This spike in worst-case execution time is due
to caching misses and would be avoided when lock operations
are performed on a warm cache. Observe from Figure~\ref{fig:uncontested_acquisition},
the maximum cost of BPL is only 50 cycles over FL.
The overhead of BPL is
offset by the gains it achieves in terms of providing bounded delay to
all waiters and reduced priority inversions compared to a FIFO
lock. BPL's cost is further mitigated in cases where it applies to a
big kernel lock or relatively large critical section. 

\subsubsection{Working Lock Performance in an RTOS}

\begin{figure*}
  \centering
  \includegraphics[width=0.3\textwidth]{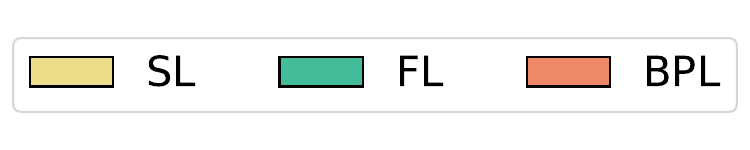}
  
  \vspace{0.3cm}
  
  \begin{subfigure}[b]{0.49\textwidth}
    \centering
    \includegraphics[width=\linewidth]{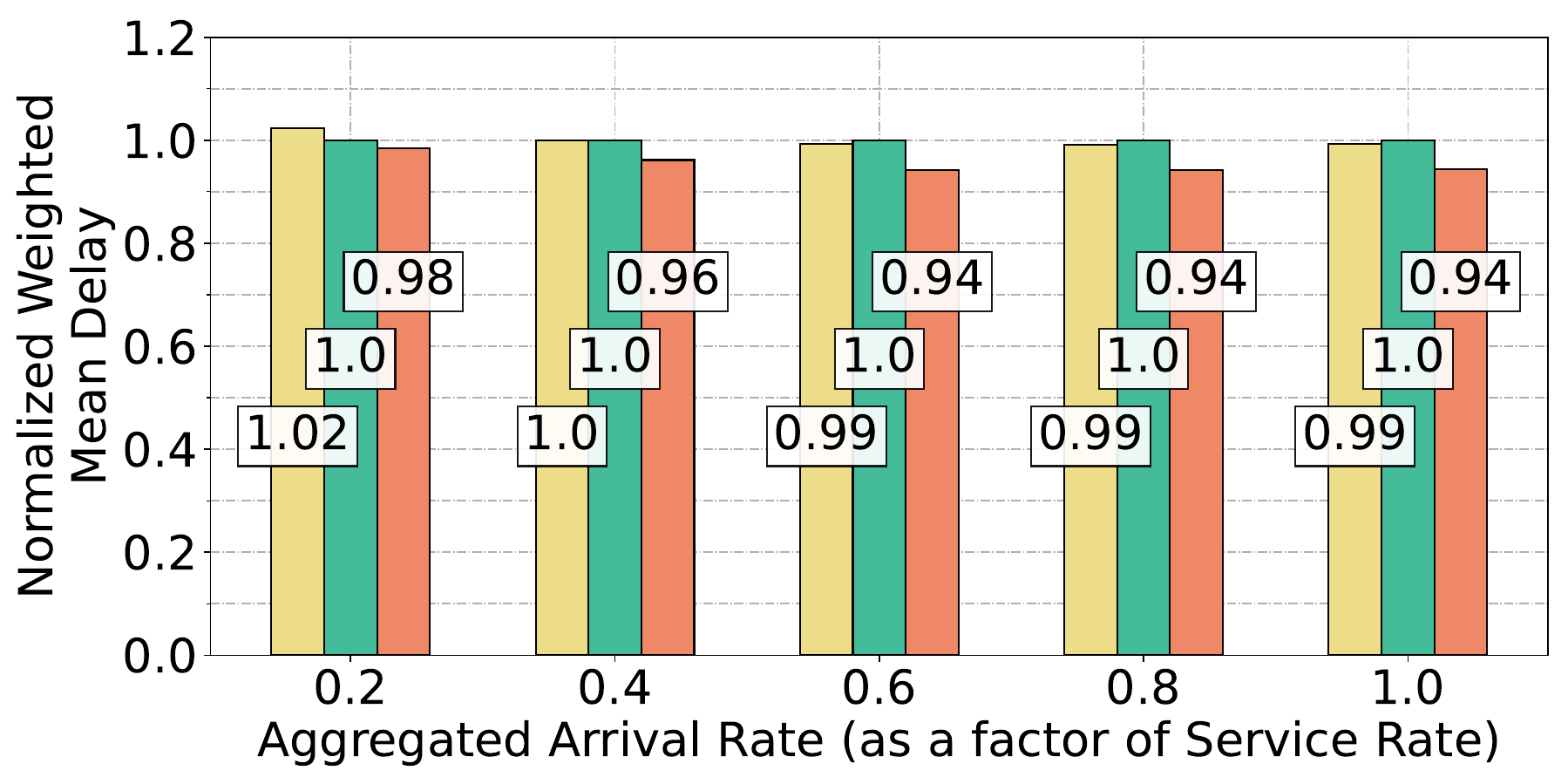}
    \caption{Weighted Mean Delay\newline \phantom{aaa\,}(equal arrival rates)}
    \label{fig:dx_wavg_eq_underload}
  \end{subfigure}
  \hfill
  \begin{subfigure}[b]{0.49\textwidth}
    \centering
    \includegraphics[width=\linewidth]{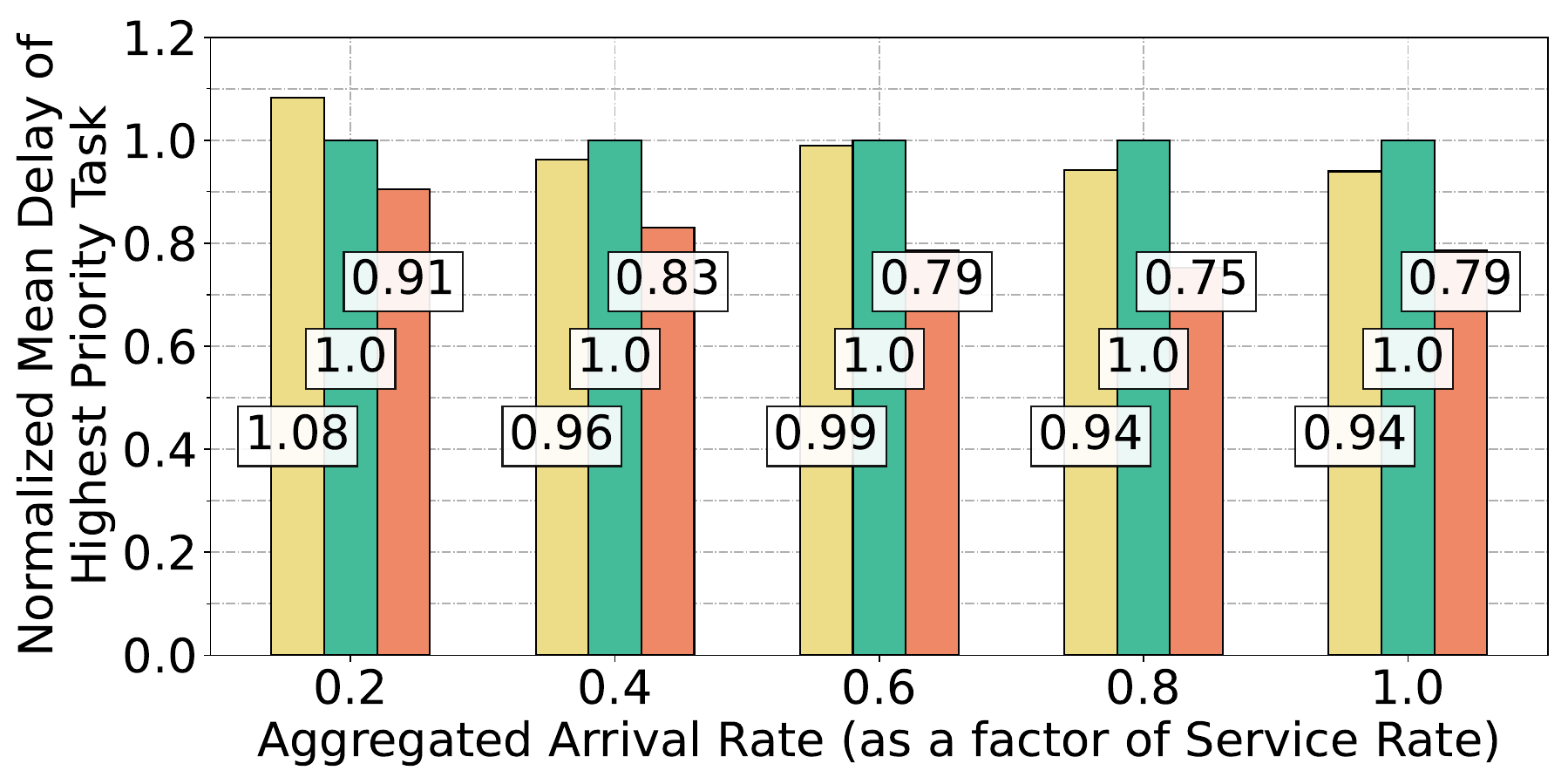}
    \caption{Delay of Highest Priority Task\newline \phantom{aaa\,}(equal arrival rates)}
    \label{fig:dx_hpavg_eq_underload}
  \end{subfigure}
  
  \vspace{0.3cm}
  
  \begin{subfigure}[b]{0.49\textwidth}
    \centering
    \includegraphics[width=\linewidth]{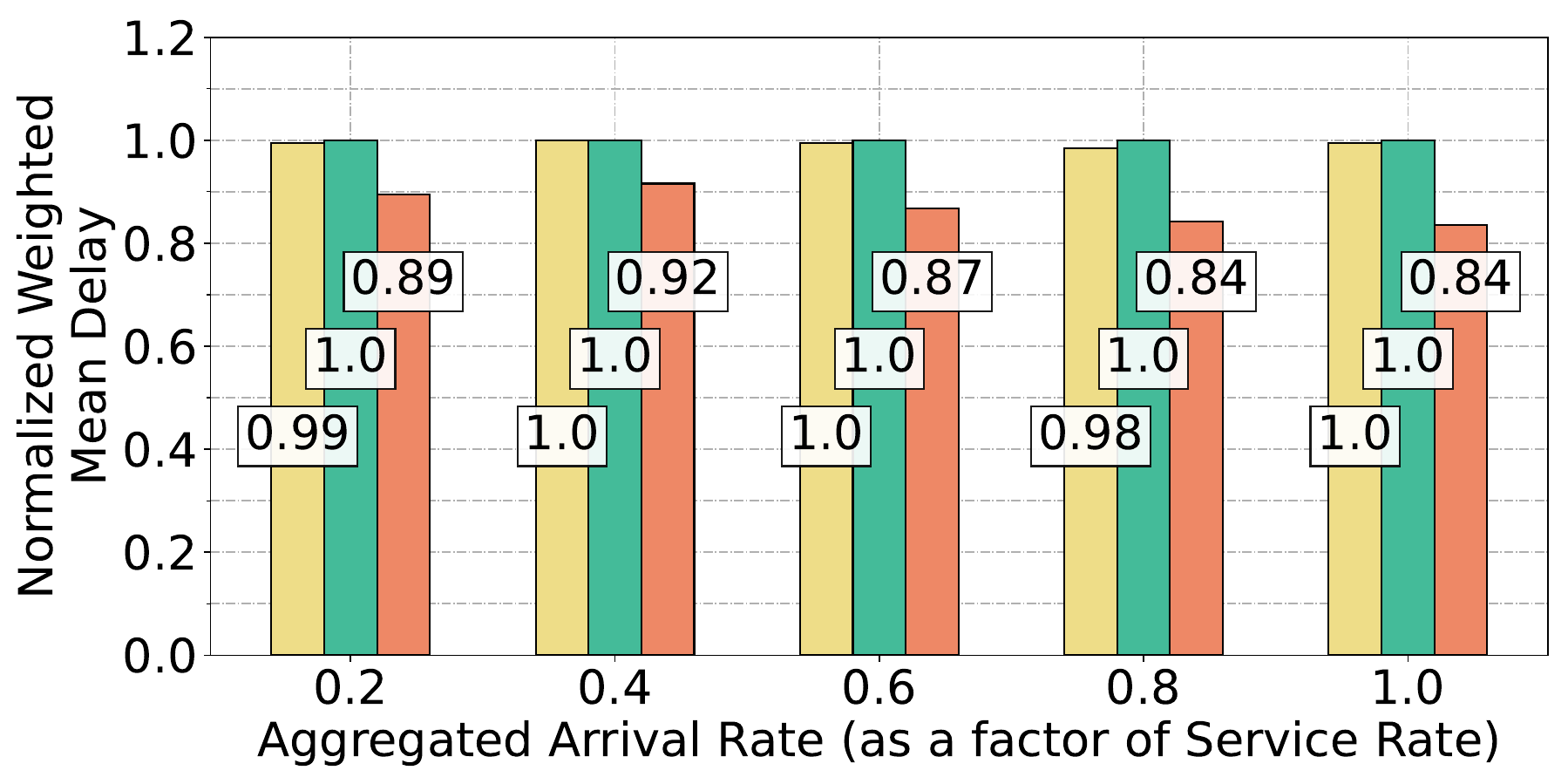}
    \caption{Weighted Mean Delay\newline \phantom{aaa\,}(higher priority tasks have lower arrival rates)}
    \label{fig:dx_wavg_hpless_underload}
  \end{subfigure}
  \hfill
  \begin{subfigure}[b]{0.49\textwidth}
    \centering
    \includegraphics[width=\linewidth]{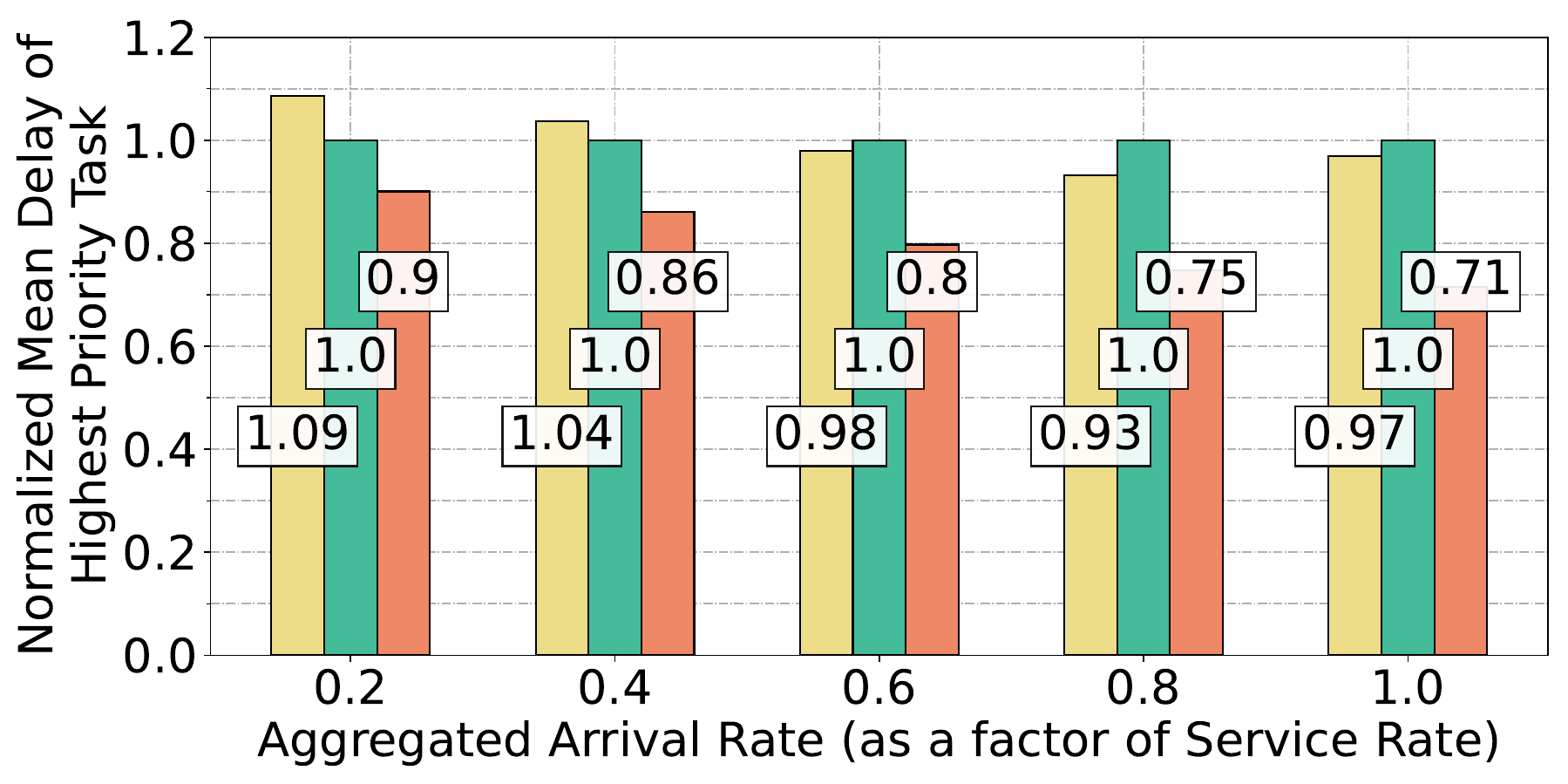}
    \caption{Delay of Highest Priority Task\newline \phantom{aaa\,}(higher priority tasks have lower arrival rates)}
    \label{fig:dx_hpavg_hpless_underload}
  \end{subfigure}
  
  \caption{Comparison of Lock Performance for 8 Cores on an RTOS}\label{fig:dx_underload}
  \Description{Comparison of Lock Performance for 8 Cores on an RTOS}
\end{figure*}

Next, all three lock implementations in Quest are compared with
independent traffic
arrival patterns having an aggregated arrival rate, $\lambda_{agg}$
described in Section~\ref{sect:sim-BPL}. We use independent rather
than a bursty arrival process, as this is more generic, and also because it is more difficult to ensure
separate tasks on different cores arrive at exactly the same time to
contend for a shared resource (here, the kernel). $\lambda_{agg}$ is
varied from $0.2\mu$ to $1.0\mu$ for two cases: (1) where tasks on all cores
have equal arrival rates, irrespective of priority, and (2) where
higher priority tasks have lower arrival rates. All tasks have
distinct priorities across the 8 cores of the machine.

Note that
while we use only $m$ tasks and vary their arrival rate, this is equivalent to
testing with multiple tasks assigned to each core because we only consider
non-preemptible waiters and lock holders.
In case (1),
the individual task arrival rates, $\lambda_i$, are set to
$\frac{\lambda_{agg}}{8}$, and in case (2), each task's arrival rate,
$\lambda_i=\frac{r_i{\cdot}\lambda_{agg}}{36}$, where $r_i\in[1,m]$,
with $r_i=1$ for the highest priority task, down to $r_i=8$ for the
lowest priority task.
In all cases, the length of the critical section is set to
70 $\mu s$ (i.e., the service rate is 14 kHz) which is
the median time to service a {\tt
usb\_write} system call guarded by a kernel lock in Quest. Each task is
assigned to a VCPU, with a utilization of 99\%, i.e., the VCPU receives
an execution budget of 99 time units for every period of 100 time units. A total of
80,000 system call requests are made across all tasks, each attempting
to acquire the kernel lock.

Figures~\ref{fig:dx_wavg_eq_underload} and~\ref{fig:dx_wavg_hpless_underload}
show the weighted mean delays, normalized against FIFO locking, for
cases (1) and (2), respectively. While BPL generally achieves
better performance as the aggregated arrival rate is
increased in both cases (1) and (2), it performs particularly well in the latter. Its normalized weighted mean delay is reduced up to 16\% in Figure~\ref{fig:dx_wavg_hpless_underload}, because BPL is able
to give precedence to higher priority arrivals in the same batch. As
higher priority arrivals are less frequent, they are more likely to be
differentiated from lower priority tasks when they do arrive, using
BPL. Thus, while BPL generally outperforms FL, it is particularly effective in
cases where a shared resource is accessed less frequently by higher priority
tasks. 
The normalized delays of the highest priority task shown in
Figures~\ref{fig:dx_hpavg_eq_underload} and~\ref{fig:dx_hpavg_hpless_underload}
more clearly demonstrate the benefits of BPL on a real-time system. 
It must be noted that, BPL does 
not have a detrimental impact on the overall delay of all tasks. 
These results confirm the properties of BPL in a practical
system. Note that SL rarely does worse than FL but if we were to go to
higher contention rates it would not guarantee bounded delay, similar
to PL.

\section{Related Work}\label{sect:rel_work}
There is a wealth of literature addressing the multiprocessor
synchronization problem in real-time systems, including the seminal
works by Rajkumar et al. on priority ceiling protocols in the late
eighties~\cite{dpcp} and early nineties~\cite{mpcp}. A thorough
summary of the literature is provided by Brandenburg~\cite{bburg_lockreview}. 
Several of the recent real-time synchronization
approaches either focus on supporting suspension~\cite{omlp}, or priority
assignment to spinlocks~\cite{ospa}. However, these approaches are not directly
applicable to our focus on big kernel locks, because scheduling decisions are
made only after acquiring exclusive access to the kernel. Here, we summarize a
few relevant locks. 

The MCS lock~\cite{mcs_lock} follows a simple FIFO-ordered enqueue
process, where tasks spin on core-local memory to avoid cache-related
delays. The simplicity of MCS locks has resulted in them being
implemented in Linux~\cite{mcs_linux}. Other variants of FIFO locks
include those that allow attempts, timeouts or aborts, and are often
referred to as
trylocks~\cite{preempt_ticket,abort_mutex,fifo_trylock_4slot,try_mcs}.

Markatos~\cite{markatos_plock} proposed a simple priority-ordered lock, where
waiters append themselves to a queue, which is traversed by the releasing task
to determine the next holder. This attributes more complexity to the lock
releasing task, effectively extending its critical section overhead.  The CLH
lock, independently developed by Craig~\cite{clh_c} and Magnusson et
al.~\cite{clh_lh} follows a similar approach to implement FIFO or
priority-ordering while providing cache benefits by allowing spinning on
core-local memory. Johnson and Harathi presented a prioritized spinlock, where a
task attempting to acquire the lock traverses the queue to insert itself
according to priority. This approach yields a O(1) release
time~\cite{o1_spinlock}.

In the aforementioned works, there is an increase in complexity to
maintain scalability. Huang and Jayanti presented a formal definition
for priority mutual exclusion~\cite{priority_mutual_exclusion}, and
highlighted some of the issues with previous work. Their algorithm
uses as many wait queues as there are priority levels in the system,
in an approach similar to multi-level queue scheduling.  Tasks of
equal priority reside in the same wait queue, and the lock releasing
task iterates through the heads of each wait queue in priority order
to find the next holder. Even though this method is somewhat simpler
than previous approaches, it is not applicable to systems with a
variable number of priority levels (e.g., when priorities are based on
task deadlines).
Fuerst~\cite{fuerst_plock} proposed a priority lock that operates for
a maximum of 64 priority levels. The priority of a task is represented
using a 64-bit bitvector, replacing the need to use a priority
queue or other data structure such as a heap. After registering itself
using the bitvector, a task checks if there are any higher priority
waiters. If there are other waiters, the task spins, otherwise it
attempts to set itself as the holder using a compare-and-exchange
operation on a global pointer. BPL ensures a FIFO ordering among tasks
of equal priority, and also supports more than 64 priority levels if
so desired. 

The BPL approach differs from others in that it attempts to bound the
lock delay, while emphasizing task importance among lock waiters that
are clustered in the same batch. BPL exhibits low overheads for
uncontended lock acquisition. Similarly, the lock release overheads
experience the same O(1) time bound as that proposed by Johnson and
Harathi. We consider BPL more beneficial to real-time operating
systems guarded by big kernel locks. In such cases, the increase in
overhead compared to a simple spinlock is offset by the longer
duration of a critical section, compared to when using fine-grained
locks. Peters et al. show with empirical analysis that the use of a big
kernel lock is suitable in a microkernel~\cite{ukernel_bkl_fine}.

\section{Conclusions}\label{sect:conclusion}
This paper describes the design and implementation of the Batched Priority
Locking (BPL) algorithm. BPL uses batching to bound the time that any task waits
to acquire a spinlock, while also considering the priority of waiters within the
same batch to order lock acquisition. Simulation and empirical analysis confirms
that BPL reduces priority inversions compared to FIFO-ordered locks, while still
retaining the same bounded delay for lock access. In comparison, a strictly
priority-ordered lock reduces priority inversions at the cost of potentially
unbounded delay to lower priority tasks. This is especially prevalent in cases
of high lock contention. For real-time systems, locks that enforce bounded wait
delays are necessary. BPL has benefits over FIFO when there are
quality-of-service gains for reducing delay for higher priority lock contenders.
The implementation overheads of BPL are offset by the benefits gained,
especially in cases where critical sections are relatively long compared to the
cost of the lock acquisition and release. This is because multiple tasks are
grouped into the same batch, where they are ordered according to their
priority. Only when tasks are partitioned into batch sizes of 1, does the
algorithm degrade to FIFO ordering. 

For short-lived critical sections e.g., to do an atomic increment
such as \verb|x++|, the costs of using BPL reduce the benefits of the locking
mechanism. In our case, we see value for accessing a big kernel lock that could
have relatively long-lived critical control paths, and BPL minimizes priority
inversions in such cases, reducing latency for higher priority tasks. Here, the
lower delay for higher priority tasks yields reward in terms of being able to
make sure computations meet a minimum (mandatory) level of service, while
allowing for optional computations to additionally take place, time permitting,
to improve the quality of the result. Imprecise
computations~\cite{Liu:91,Liu:94} is an example that could benefit from BPL, as
applied to refinement tasks such as numerical integration, state estimation and
so forth.

Future work will consider the scalability limits of BPL on a single
kernel, versus an alternative approach where large core counts are
partitioned across multiple guest systems. For example, a single
system with $m$ cores contending for a single lock could be compared
to a partitioning hypervisor system, hosting $n$ guests each with
$m/n$ cores contending for a shared lock. How to assign tasks to
the same guest under such a scenario will also be investigated.


\bibliographystyle{ACM-Reference-Format}
\bibliography{references}

\end{document}